\documentclass[11pt,twoside,a4paper,american]{article}
\newif\ifdraft \global\drafttrue

\usepackage{microtype}
\usepackage[T1]{fontenc}
\usepackage[utf8]{inputenc}
\usepackage{longtable}

\usepackage{a4wide}
\usepackage{times}
\usepackage{theorem}
\usepackage[colorlinks=true,hyperindex=true]{hyperref}
\usepackage{fancyhdr}
\usepackage{latexsym}
\usepackage{amsmath,amsfonts,amssymb}
\usepackage{bbm}
\usepackage{bm}
\usepackage{enumerate}
\usepackage{cancel}
\usepackage{stmaryrd}
\usepackage{color,xcolor}
\usepackage{xspace}
\ifdraft
\fi
\colorlet{darkblue}{blue!50!black}

\hypersetup{
    colorlinks,%
    citecolor=darkblue,%
    filecolor=red,%
    linkcolor=darkblue,%
    urlcolor=magenta,%
    pdfnewwindow=true,%
    pdfstartview={FitH}
}

\newcounter{smallarabics}

\newcounter{smallroman}

\newcommand{\ben}{\begin{enumerate}[(1)]}
\newcommand{\een}{\end{enumerate}}

\newtheorem{theorem}{Theorem}[section]
\newtheorem{proposition}[theorem]{Proposition}
\newtheorem{lemma}[theorem]{Lemma}
\newtheorem{corollary}[theorem]{Corollary}
\newtheorem{definition}[theorem]{Definition}

{\theorembodyfont{\upshape}
\newtheorem{remark}[theorem]{Remark}
\newtheorem{example}[theorem]{Example}
\newtheorem{example*}{Example}
}

\def\bep{\begin{proposition}}
\def\eep{\end{proposition}}

\def\bel{\begin{lemma}}
\def\eel{\end{lemma}}
\def\bet{\begin{theorem}}
\def\eet{\end{theorem}}
\def\bed{\begin{definition}}
\def\eed{\end{definition}}
\def\bec{\begin{corollary}}
\def\eec{\end{corollary}}

\newcommand{\R}{{\mathbb R}}
\newcommand{\Z}{{\mathbb Z}}
\newcommand{\IP}{{\mathbb P}}
\newcommand{\Q}{{\mathbb Q}}

\newcommand{\I}{{\mathbb I}}

\newcommand{\argdot}{{\bm \cdot }}
\newcommand{\bdot}[1]{{\overset{\bm .}{#1}}}

\def\rr{{\mathbb R}}
\def\zz{{\mathbb Z}}
\def\cc{{\mathbb C}}

\def\Z{{\mathbb Z}}

\def\textsl{{}}

\def\Re{\mathrm{Re}\,}

\newcommand{\Id}{\mathop{\mathrm{Id}}\nolimits}

\def\c0inf{C_0^\infty}

\def\cV{{\cal  V}}

\def\cA{{\cal A}}
\def\PP{\mathcal{P}}
\def\EE{\mathcal{E}}
\def\QQ{\mathbb{Q}}
\def\II{\mathcal{I}}

\newcommand{\beq}{\begin{equation}}
\newcommand{\eeq}{\end{equation}}
\newcommand{\bear}[1]{\begin{array}{#1}}
\newcommand{\ear}{\end{array}}

\newcommand{\eee}{\mathrm{e}}

\renewcommand{\d}{\mathrm{d}}

\newcommand{\ep}{\mathrm{ep}}

\def\qed{$\Box$\medskip}

\def\cP{{\cal P}}

\def\cA{{\cal A}}

\def\bar{\overline}

\def\12{\frac{1}{2}}

\def\dd{{\textup d}}

\def\d{\mathrm{d}}

\def\Ent{\mathrm{Ent}}

\def\P{\mathbb P}

\def\ie{\textit{i.e., }}

\def\wP{{\widehat{\mathbb P}}}

\def\e{{\epsilon}}

\newcommand{\aA}{{\cal A}}
\newcommand{\BB}{{\cal B}}

\newcommand{\FF}{{\cal F}}
\newcommand{\GG}{{\cal G}}

\newcommand{\VV}{{\cal V}}

\newcommand{\N}{{\mathbb N}}

\def\fp{\mathfrak{p}}

\newcommand{\ssss}{{\mathfrak s}}
\newcommand{\hhhh}{{\mathfrak h}}

\newcommand{\tA}{\hyperlink{thm.A}{\textbf{A}}\xspace}
\newcommand{\tB}{\hyperlink{thm.B}{\textbf{B}}\xspace}

\numberwithin{equation}{section}
\begin{document}
\def\today{}
\title{Fluctuation Theorem and Thermodynamic Formalism}
\author{ Noé~Cuneo$^{1,3}$, Vojkan Jak\v{s}i\'c$^{1}$, Claude-Alain  Pillet$^{2}$,
Armen~Shirikyan$^{1,3,4}$
\\ \\ \\
$^1$Department of Mathematics and Statistics, McGill University\\ 
805 Sherbrooke Street West, Montreal, QC, H3A 2K6, Canada 
\\ \\
$^2$Aix Marseille Univ, Université de Toulon, CNRS, CPT, Marseille, France  
\\ \\
$^3$Department of Mathematics, University of Cergy--Pontoise, CNRS UMR 8088\\
2 avenue Adolphe Chauvin, 95302 Cergy--Pontoise, France
\\ \\
$^4$Centre de Recherches Math\'ematiques, CNRS UMI 3457\\
Universit\'e de Montr\'eal, Montr\'eal,  QC, H3C 3J7, Canada}

\maketitle

{\small\textbf{Abstract.} 
We study the Fluctuation Theorem (FT) for entropy production in chaotic
discrete-time  dynamical systems on compact metric spaces, and extend it to
empirical measures, all continuous potentials, and all weak Gibbs states. In
particular, we establish the FT in the phase transition regime. These results
hold under minimal chaoticity assumptions (expansiveness and specification) and
require no ergodicity conditions. They are also valid  for systems that are not
necessarily invertible and involutions other than time reversal. Further
extensions involve asymptotically additive potential sequences and the
corresponding weak Gibbs measures. The generality of these results allows to
view the FT as a structural facet of the thermodynamic formalism of dynamical
systems.

\bigskip


}

\tableofcontents
\setcounter{section}{-1}


\bigskip

{\bf Note added in Feb.~2026}: most of the content of the present preprint has been published, in extended form, in the monograph \cite{CJPS_phys}. In the meantime, it was proved in \cite{cuneo_asympt_2019} that asymptotically additive potentials are equivalent to additive ones. This resolves the open question raised in Remark~\ref{rem:openQaa} at the end of the present preprint and renders part of the proofs below superfluous, although technically correct.

\section{Introduction}
\setlength{\parskip}{4pt}

This work concerns the mathematical theory of the so-called Fluctuation Relation
(FR) and Fluctuation Theorem (FT) in the setting of discrete-time continuous
dynamical systems on compact  metric spaces. The FR is a universal property of
the statistics of entropy production linked to time-reversal and the FT refers
to a related Large  Deviation Principle (LDP).

The discovery of the FR goes back to numerical experiments on the probability of
violation of the 2$^\mathrm{nd}$ Law of Thermodynamics~\cite{ECM-1993} and
associated theoretical works~\cite{ES-1994,GC-1995,GC-1995a,gallavotti-1995} in
the early 90's. In particular, the first formulation and mathematical proof of
the FT were given in~\cite{GC-1995} in the context of Anosov diffeomorphisms of
compact Riemannian manifolds. Further steps in the mathematical development of
the subject were taken in~\cite{kurchan-1998,LS-1999,maes-1999,ruelle-1999}.
These discoveries generated an enormous body of theoretical, numerical and
experimental works which have fundamentally improved our understanding of
non-equilibrium physics, with applications extending to chemistry and biology.
For a  review of  these historical developments we refer the reader
to~\cite{ES-2002,JQQ2004,gaspard-2005,RM-2007,JPRB-2011} and to the forthcoming 
review articles~\cite{CJPS_phys,JNPPS}; see also Example~\ref{e3} below.  The
general mathematical structure and interpretation of the FR and FT from a modern
point of view is briefly discussed in Section~\ref{s1};
see~\cite{CJPS_phys,JNPPS} for additional information.

We shall consider dynamical systems $(M,\varphi)$, where~$M$ is a compact metric
space and $\varphi:M\to M$ is a continuous map. This is precisely the setting in
which the FR and FT were initially discovered. We shall assume that~$(M,
\varphi)$ is chaotic in the  sense that~$\varphi$ is expansive and satisfies
Bowen's specification property.\footnote{See~\cite[Section~18.3.c]{KH1995} and
Remark~\ref{rBSP} below. This assumption is made only for simplicity of
exposition---all our results hold under a weaker assumption, see
Section~\ref{s2.1} for a precise statement.} In this context, the
thermodynamic formalism singles out two sets of measures associated with
$(M,\varphi)$: equilibrium states and Gibbs states. Accordingly, we shall prove
two conceptually independent sets of results. The first one deals with the
entropy production of equilibrium states defined through their approach via
periodic orbits (see Proposition~\ref{p2.8}).
We shall refer to this set of results as the \textit{Periodic
Orbits Fluctuation Principle\/} (POFP). The second one concerns
extensions of the  classical Gibbs type FR and FT to weak Gibbs measures. We
shall refer to this second set of results as the \textit{Gibbs Fluctuation
Principle\/} (GFP). Together, the POFP and GFP constitute  a technical and
conceptual extension of the previously known results on FR and FT. For example,
the POFP  holds for any continuous potential, and, more generally, for any
asymptotically additive (not necessarily continuous) potential sequence. In the
case of two-sided  subshifts of finite type, the GFP holds for all Gibbs states
(translation invariant or not) of  any summable interaction~$\varPhi$. The first
result is new while the second (and only in part) was known to hold for
interactions~$\varPhi$ satisfying Bowen's regularity assumption  and, in
particular, admitting a unique Gibbs state; see Example~\ref{e1} below.

In the usual sense, the FT and FR are related to time reversal and require the
map $\varphi$ to be invertible. In this paper, we shall consider also
involutions (other than time reversal) that do not require the invertibility of
$\varphi$, and show that the FT and FR naturally extend to that case (see
Section~\ref{s2.3} for precise definitions).

We now describe some of our typical results. For simplicity, we consider only
the invertible case in the remaining part of this introduction. We assume that
$\varphi$ is a homeomorphism, and introduce the following notion of 
\textit{reversal map\/}: there is a continuous map $\theta:M\to M$ such
that
\begin{equation}\label{1.2}
\theta\circ\theta=\Id_M, \quad
\varphi^{-1}=\theta\circ\varphi\circ\theta,
\end{equation} 
where $\Id_M$ stands for the identity map on~$M$. Although in the main text
of the paper our results are stated and proven in the general setting of the
asymptotically additive thermodynamic formalism,  we  shall start with the
familiar additive setting before turning to that level of
generality.\footnote{We shall freely use the standard notions of the usual
thermodynamic formalism~\cite{ruelle2004,walters1982}. For the asymptotically
additive extensions see~\cite{Bar2011} and  Section~\ref{s1b}.} We  fix an
arbitrary continuous function\footnote{Following the usual terminology, we shall
often refer to $G$ as a \textit{potential\/}. The adjective \textit{additive\/} 
refers to the property $S_{n+m}G= S_mG + S_nG\circ \varphi^m$ of the sequence 
$\{S_nG\}$.}
$G:M\to\R$,  and
set
\[
S_nG=G+G\circ\varphi+\cdots+G\circ\varphi^{n-1}.
\] 

We start with the POFP. Denote by $M_n$ the set of $n$-periodic points of
$\varphi$. Under our assumptions $M_n$ is non-empty, finite, invariant under
$\theta$, and $\bigcup_n M_n$ is dense in $M$. We define a family of
probability measures on~$M$ by
\begin{equation}\label{2.8add}
\IP_n(\dd y)=Z_n^{-1}\sum_{x\in M_n}\eee^{S_nG(x)}\delta_x(\dd y), \quad
Z_n=Z_n(G)=\sum_{x\in M_n}\eee^{S_nG(x)},
\end{equation}
where $n\ge1$ and $\delta_x$ denotes the Dirac mass at $x$. Let 
\begin{equation}\label{1.5}
\widehat\IP_n(\dd y)=(\IP_n\circ\theta)(\dd y)
=Z_n^{-1}\sum_{x\in M_n}\eee^{S_nG\circ\theta(x)}\delta_x(\dd y).
\end{equation}
The measures $\widehat\IP_n$ and~$\IP_n$ are absolutely continuous with respect to each other, and the logarithm of the corresponding density is given by
\begin{equation*}
\log\frac{\dd\IP_n}{\dd\widehat\IP_n}(x)=S_n\sigma(x)\quad\text{for }x\in M_n,
\end{equation*}
where we write 
\begin{equation}\label{eq:defsigmaGG}
	\sigma=G-G\circ\theta
\end{equation}
for the \textit{entropy production observable\/}. Any weak limit point $\IP$ of the
sequence $\P_n$ is an equilibrium measure for $G$ (see Proposition~\ref{p2.8}).
Note that if $\P_{n_k}\rightharpoonup\P$, then
$\wP_{n_k}\rightharpoonup\wP=\P\circ \theta$. The mathematical statement of the
POFP  is the Large Deviation Principle (LDP) for the empirical measures and
ergodic averages of~$\sigma$ with respect to~$\P_n$. Its interpretation, on
which we shall elaborate in Section~\ref{s1}, quantifies the separation 
between~$\wP_{n_k}$ and~$\P_{n_k}$, as these sequences of measures approach
their  limits~$\wP$ and~$\P$.

Let $\PP(M)$ be the set of all probability measures on~$M$ endowed with the
topology of weak convergence. The following theorem summarizes the POFP.

\medskip
\hypertarget{thm.A}{\textbf{Theorem~A.}}
{\itshape For any continuous function $G:M\to\R$, the following assertions hold.
\begin{description}
\item[Large deviations.] 
There is a lower semicontinuous  function $\I:\PP(M)\to[0,+\infty]$  such that the sequence of empirical measures 
\beq
\mu_n^x=\frac1n\sum_{k=0}^{n-1}\delta_{\varphi^k(x)}
\label{oc-me}
\eeq
under the law~$\IP_n$  satisfies the LDP with the rate function~$\I$.

\item[Fluctuation theorem.] 
The sequence  $\frac 1nS_n \sigma$ under the law~$\IP_n$ satisfies the LDP with a rate function~$I$ given by the contraction of\/~$\I:$
\begin{equation} \label{1.7}
I(s)=\inf\biggl\{\I(\Q):\Q\in\PP(M),\int_M\sigma\,\dd\Q=s\biggr\}. 
\end{equation}
\item[Fluctuation relations.] 
The rate functions\/~$\I$ and~$I$ satisfy the relations
\begin{gather}
\I(\widehat \Q)=\I(\Q)+\int_M\sigma\,\dd\Q,\label{1.8}\\[4pt]
I(-s)=I(s)+s,\label{1.8b}
\end{gather}
where $\Q\in\PP(M)$, $s\in\R$ are arbitrary, and $\widehat\Q=\Q\circ\theta$.
\end{description}
}
\begin{remark}
The importance of periodic orbits for the study of chaotic dynamics in the modern
theory of dynamical systems goes back to seminal works of
Bowen~\cite{bowen-1970} and Manning~\cite{manning-1971}. In the context of the
FT and FR, periodic orbits played an important role in the early numerical
works~\cite{ECM-1993}. Ruelle's proof of the Gallavotti--Cohen fluctuation
theorem for Anosov diffeomorphisms~\cite{ruelle-1999} was technically centered
around periodic orbits. Further insights were obtained in~\cite{MV-2003} where,
following the general scheme of~\cite{LS-1999,maes-1999}, the pairs $(\P_n,
\wP_n)$ and the entropy production observable~$\sigma$ were introduced, and the
transient fluctuation relation was discussed. The work~\cite{MV-2003} primarily
concerned Gibbs type FT for Bowen-regular potentials~$G$,  and we shall comment
further on it in Example~\ref{e3} below.
\end{remark}

\begin{remark}
In the early physics literature on the subject, the fluctuation relation was 
usually stated as the universal large-$n$ asymptotics
\[
\frac{\mathrm{Prob}\{S_n\sigma=-sn\}}{\mathrm{Prob}\{S_n\sigma=sn\}}\simeq\eee^{-sn},
\]
where the quotient on the left hand side should be interpreted as a
Radon-Nikodym derivative. It was first emphasized by Gallavotti and
Cohen~\cite[Section~7]{GC-1995} that a proper mathematical statement of this
fact was the FT, \ie the LDP satisfied by $n^{-1}S_n\sigma$,
together with relation~\eqref{1.8b} for the associated rate function.
In some sense, relation~\eqref{1.8} satisfied by the rate function $\I$ governing
the large deviations of the measures \eqref{oc-me} is
more general. In fact, it implies~\eqref{1.8b} for the systems
considered in this paper (see the last paragraph in the proof of 
Theorem~\ref{t1.14}). The FR \eqref{1.8} and its formal connection with~\eqref{1.8b}
was already noticed by Bodineau and Lefevere~\cite[Section~3.4]{BoLe-2008}
and by Barato and Chetrite in~\cite[Section~5]{BC-2015} in 
somewhat different contexts.
 We refer the reader to~\cite{CJPS_phys} for further
discussion of this point.
\end{remark}

We now turn to the GFP. We shall assume that~$\IP$ is a  weak Gibbs measure for
some potential~$G\in C(M)$\footnote{$C(M)/B(M)$ denotes the usual Banach space
of real-valued continuous/bounded Borel functions on $M$.}  (see
Definition~\ref{d1.16} with $G_n=S_nG$).

\medskip
\hypertarget{thm.B}{\textbf{Theorem~B.}}
{\itshape Let~$\IP$ be a weak  Gibbs measure for a potential $G\in C(M)$. Then the
three assertions of Theorem~\textbf{A} remain valid if we replace~$\IP_n$
with~$\IP$. }

\begin{remark}\label{r2}
On a technical level the key point of Theorems~\tA and~\tB is the 
LDP for the empirical measures~\eqref{oc-me}, while the remaining properties are
easy consequences of it.
The respective rate functions in Theorems~\tA
and~\tB coincide.\footnote{This fact is related to the  \textit{principle of
regular entropic fluctuations\/} introduced in \cite{JPRB-2011};  see
Section~\ref{s1.5}.} The FR for the rate function~$\I$ can be derived from an
explicit formula (see \eqref{4.012first}), while  the FR for~$I$ is implied by the contraction
relation~\eqref{1.7}. For a more conceptual derivation of the FR for~$I$ see
Section~\ref{s1.1}.
\end{remark}

\begin{remark}\label{r3}
In the additive setup discussed here, the LDP for the empirical measures~\eqref{oc-me} in Theorems~\tA 
and~\tB is known \cite{comman-2009,PS-2018} (see Remarks~\ref{rem:additivecaseknown} and \ref{rem:commentt44} below).
The proof that we provide applies to asymptotically additive potentials, 
and for these the result seems to be new. We will actually state and
prove a theorem in Section~\ref{s:abstractLevel3} which implies the LDP for the empirical measures
of both Theorem~\tA and Theorem~\tB.
The proof of the LDP involves, as
usual, two steps: the LD upper bound, which is a simple consequence of the
existence of the topological pressure (see Propositions~\ref{Thm:USCh}
and~\ref{p1.10}), and the LD lower bound, which is more involved. A prototype of
our argument appeared in the proofs of Theorem~3.1 in
F\"ollmer--Orey~\cite{FO-1988} and Theorem~2.1 in Orey--Pelikan~\cite{OP-1988},
in which the Shannon--McMillan--Breiman (SMB) theorem is used to derive the LD
lower bound for Gibbs states of~$\zz^d$ spin systems (see also \cite{comets_86,Olla_88}).  By using Markov
partitions, the same result was established for transitive Anosov
diffeomorphisms on compact manifolds~\cite{OP-1989}.\footnote{It is interesting
to note that  this result and the contraction principle immediately yield the
Gallavotti--Cohen FT.}  In our context, the SMB theorem is naturally replaced by
its dynamical systems counterpart, the Brin--Katok local entropy
formula~\cite{BK-1983}. The rest of our argument is related  to the
papers~\cite{young90,EKW-1994,PS-2005}. Although there
the Brin--Katok theorem is not used directly, the key estimates entering
Proposition~4.2 in~\cite{EKW-1994} and Proposition~3.1 in~\cite{PS-2005} are
also important ingredients in the proof of the Brin--Katok formula and can be
traced back to another work of Katok~\cite[Theorem~1.1]{katok-1980}.
\end{remark}

\begin{remark}\label{r4} 
As Remarks~\ref{r2} and~\ref{r3} indicate, on the technical level
Theorems~\tA and~\tB are closely related. We have separated them
for historical reasons, for reasons of interpretation, and due to the role
played by the specification property in the proofs. Regarding the first two points, see
Example~\ref{e3} below and Section~\ref{s1}. Regarding the third one, in
Theorem~\tB, the specification is only needed to allow the use of
Proposition~\ref{Prop:EntropyDensity}, whereas in Theorem~\tA a weak form
of specification is crucial also in the proof of the lower bound of the LDP.
There are alternative assumptions under which the conclusions of
Proposition~\ref{Prop:EntropyDensity} can be established. For example, Pfister
and Sullivan~\cite{PS-2005} prove it for dynamical systems with the so-called
$g$-product property and apply it to $\beta$-shifts. Hence, Theorem~\tB
holds in that setting. See also~\cite{comman-2017} for other criteria
ensuring the validity of the conclusion of
Proposition~\ref{Prop:EntropyDensity}.
\end{remark}

Before turning to the asymptotically  additive setting, we briefly discuss
several prototypical additive examples; see also Example~\ref{interval} in
Section~\ref{s2.3}. For details and  additional examples, we refer the reader
to the accompanying review article \cite{CJPS_phys}.

\begin{example}[Two-sided subshift of finite type]\label{e1}
Let\,\footnote{Here and in the sequel $\llbracket 1,\ell\rrbracket=[1, \ell]\cap
\zz$.} $\aA=\llbracket 1,\ell\rrbracket$ be a finite alphabet with the 
discrete metric and let $\Omega=\aA^\Z$ be the product space of two-sided 
sequences endowed with the usual metric
\begin{equation*}
\Omega\times\Omega\ni(x,y)\mapsto
d(x, y)=2^{-\min\{j\in\Z_+\,:\,x_j\ne y_j\text{ or }
x_{-j}\ne y_{-j}\}}.
\end{equation*}
The left shift $\varphi:\Omega\to\Omega$ defined by $\varphi(x)_j=x_{j+1}$
is obviously an expansive homeomorphism. We assume that $(M, \varphi)$ is  a
subshift of finite type: given an $\ell\times\ell$ matrix $A=[A_{ij}]$ with
entries $A_{ij}\in\{0,1\}$ and such that, for some $m\geq 1$, all entries of the 
matrix $A^m$ are strictly positive, one sets
\[
M=\{x=(x_j)_{j\in\Z}\in\Omega: A_{x_jx_{j+1}}=1\text{ for all }j\in\Z\}. 
\]
In this case, $(M,\varphi)$ is topologically mixing and satisfies Bowen's
specification property. Let $p$ be an involutive permutation of $\cA$ and set
$\theta(x)_j=p(x_{-j})$. Then~$\theta$ is a homeomorphism of $\Omega$
satisfying~\eqref{1.2}. Thus, if in addition~$\theta$ preserves $M$, then it is
a reversal of~$(M,\varphi)$. For a subshift of finite type, that is the case
whenever the permutation matrix~$P$ associated to the map $p$ satisfies $PAP=A^T$. 

Theorem~\tA applies to any $G\in C(M)$, and hence in situations where~$G$
exhibits phase transitions and  the set of equilibrium states for $G$ is not a
singleton. Theorem~\tA also covers the  cases where $G$ exhibits
pathological behavior from the phase transition point of view; 
see~\cite[Section 3.17]{ruelle2004} and~\cite[Section V.2]{israel2015convexity}.
For example, if ${\cal P}=\{\P_1, \cdots, \P_n\}$ is any finite collection of
ergodic measures of the dynamical system~$(M, \varphi)$, then there exists a
potential~$G$  whose set of ergodic equilibrium states is precisely~${\cal P}$.
There is a dense set of~$G$'s in $C(M)$ with uncountably many ergodic
equilibrium states. Although such general potentials could  be considered 
non-physical, the POFP remains valid.

Regarding Theorem~\tB, consider a spin chain on $M$ defined by a  summable
translation-invariant interaction $\varPhi$. We shall follow the notation  of
the classical monograph~\cite{ruelle2004}, and assume that $\varPhi$ belongs to
the Banach space ${\cal B}$ of summable interactions introduced in Section 4.1
therein. We denote by $K_\varPhi\subset  {\cal P}(M)$ the set of all Gibbs
states for $\varPhi$. Then~$K_\varPhi$ is a closed convex set and some elements
of~$K_\varPhi$ may not be $\varphi$-invariant. The set of $\varphi$-invariant
elements of~$K_\varPhi$ is precisely the set of equilibrium  states for the
potential~$A_\varPhi$ (the contribution of one lattice site to the energy of a
configuration) defined in~\cite[Section~3.2]{ruelle2004}. If $A_\varPhi$
satisfies Bowen's regularity condition (see~\cite{bowen-1974} and
\cite[Definition 20.2.5]{KH1995}), then~$K_\varPhi$ is a singleton, but in
general~$K_\varPhi$ may have many distinct elements.  However,
any $\P\in K_\varPhi$ is a weak Gibbs measure for the
potential~$A_\varPhi$ (see \cite[Lemma 3.2]{EKW-1994}, \cite{PS-2019} and  \cite{CJPS_phys} for
details), and Theorem~\tB applies. These results extend to  $\Omega={\cal
A}^{\zz^d}$ for any $d\geq 1$.
\end{example}

\begin{example}[Uniformly hyperbolic systems]\label{e2}
Let $\Omega$ be a compact connected Riemannian manifold and
$\varphi:\Omega\to\Omega$ a $C^1$-diffeomorphism. Let $M\subset \Omega$ be a
locally maximal invariant hyperbolic set such that $\varphi|_M$ is transitive. 
Then  the map $\varphi$ is an expansive homeomorphism of $M$ satisfying Bowen's 
specification property. Hence Theorems~\tA and~\tB hold for $(M,
\varphi)$; see~\cite{bowen1975,PP-1990}.
\end{example}

\begin{example}[Anosov diffeomorphisms]\label{e3}
Continuing with the previous example, if $M=\Omega$, then $(\Omega,\varphi)$ is
a transitive Anosov system. This is the original setting in which the first FR
and FT were proven. We denote by $D(x)=|\det\varphi'(x)|$ the Jacobian
of~$\varphi$ at~$x$ and set
\[
D^{s/u}(x)=|\det(\varphi'(x)|_{E_x^{s/u}})|,
\]
where~$E_x^{s/u}$ denotes the stable/unstable tangent subspace at $x\in M$. 
The $C^1$-regularity of~$\varphi$ implies that the maps
\begin{equation} \label{vir}
x\mapsto D(x),\qquad x\mapsto D^{s/u}(x),
\end{equation}
are continuous. The potential   
\begin{equation} \label{potential}
G(x)=-\log D^{u}(x)
\end{equation}
is of particular importance~\cite{ER-1985,EP-1986}, and in the context of the FR
its relevance goes back to the pioneering work~\cite{ECM-1993}. As a special
case of Example~\ref{e2}, Theorem~\tA holds for this $G$ and any continuous
reversal $\theta$. Theorem~\tB holds for any weak Gibbs measure for $G$.

If $\varphi$ is $C^{1+\alpha}$ for some $\alpha>0$, then the maps~\eqref{vir}
are H\"older continuous and the potential~$G$ has a unique equilibrium state,
the SRB probability measure~$\P_\mathrm{srb}$. In this case, denoting by 
$\P_\mathrm{vol}$ the normalized Riemannian volume measure on $M$,    the empirical
measures~\eqref{oc-me} converge weakly to~$\P_\mathrm{srb}$ for
$\P_\mathrm{vol}$-a.e.~$x\in M$. The measure~$\P_\mathrm{srb}$
enjoys very strong ergodic properties and,
in particular, is weak Gibbs for $G$, so that Theorem~\tB applies to $\P_\mathrm{srb}$. In this case, the LDP part of Theorem~\tB goes back
to~\cite{OP-1989}. Since $\P_\mathrm{vol}$ is also a weak Gibbs measure for
$G$,\footnote{This follows from the Volume Lemma;  see~\cite[Lemma
4.7]{bowen1975} and~\cite[Lemma 20.4.2]{KH1995}.} Theorem~\tB applies to
$\P_\mathrm{vol}$ as well.
\end{example}

\begin{example}[Anosov diffeomorphisms: historical perspective]\label{exa:anosovhisto}
The original formulation of the  Gal\-lavotti--Cohen FT~\cite{GC-1995,GC-1995a}
concerns~$C^{1+\alpha}$ transitive Anosov diffeomorphisms with the additional
assumptions that the reversal map~$\theta$ is~$C^1$ and that the Riemannian metric is $\theta$-invariant. The entropy production
observable is taken to be  the phase space contraction rate
\begin{equation}
\widetilde{\sigma}(x)=-\log D(x),
\end{equation}
and the LDP concerns the time averages $n^{-1}S_n\widetilde{\sigma}$. Since 
$\theta$ is $C^1$, the tangent map $\theta'(x)$ provides an isomorphism 
between~$E_x^{s/u}$ and~$E_{\theta(x)}^{u/s}$, and
\begin{equation*}
\log D^u \circ \theta =-\log D^s \circ \varphi^{-1}.
\end{equation*}
As observed in \cite{MV-2003},  this relation gives that for some $C>0$, all $x\in M$ 
and all $n$, 
\[
|S_n\widetilde{\sigma}(x)-S_n\sigma(x)| <C,
\]
where $\sigma=G-G\circ \theta$ with $G$ given by \eqref{potential}.
Hence, under the assumptions of~\cite{GC-1995,GC-1995a}, the Gallavotti--Cohen 
FT and the FT of Theorem~\tB are identical statements.

The assumption that $\theta$ is~$C^1$ is essential for the Gallavotti--Cohen FT.
Porta~\cite{porta-2010} has exhibited examples of $C^\infty$ Anosov
diffeomorphisms on the torus~${\mathbb T}^2$ which admit continuous but not
differentiable reversals, and for which the Gallavotti--Cohen FT \textit{fails\/} 
in the sense that the LDP rate function for the averages
$n^{-1}S_n\widetilde{\sigma}$ does not satisfy the second relation
in~\eqref{1.8}. For his examples Porta also identifies the entropy production
observable $\sigma= G- G\circ \theta$, noticing that the LDP holds for it with a
rate function satisfying the FR~\eqref{1.8}.

Porta's observation was a rediscovery of an important insight of Maes and
Verbitskiy. Returning to our  general setting $(M, \varphi)$, in~\cite{MV-2003}
the entropy production observable $\sigma= G-G\circ \theta$ is introduced for an
arbitrary potential $G$, and the Gibbs FR and FT  were established for the
averages $n^{-1}S_n\sigma$ assuming that~$G$ satisfies the Bowen regularity
condition. In this case~$\P$ is again the unique equilibrium measure for~$G$ and
enjoys very strong ergodic properties.  The proofs of~\cite{MV-2003} are further
simplified in~\cite{JPRB-2011}; see Section~\ref{s1.4} below.
\end{example}

We now turn to the asymptotically additive setting.

\begin{definition}
A sequence of functions~$\GG=\{G_n\}_{n\geq 1}\subset B(M)$
is  called \textup{asymptotically additive}  if there is a sequence  
$\{G^{(k)}\}_{k\geq 1}\subset  C(M)$ such that
\begin{equation}\label{eq:defasymadd}
\lim_{k\to\infty}\limsup_{n\to\infty}\,n^{-1} 
\bigl\|G_n-S_n G^{(k)}\bigr\|_\infty=0. 	
\end{equation}
The set of all asymptotically additive sequences of functions on~$M$ is denoted
by~${\cA}(M)$, and a family $\{G^{(k)}\}_{k\geq 1} \subset C(M)$ 
satisfying~\eqref{eq:defasymadd} is called an\/ \textup{approximating sequence}
\footnote{Note that $G_n$ is not required to be continuous, but $G^{(k)}$ is. The 
notion of asymptotically additive potential was first introduced 
in~\cite{FH-2010}, and there $G_n$ is required to be continuous (see 
Section~\ref{sec:charac} for a detailed discussion of this point).} for~$\GG$.
\end{definition} 

Except in Section~\ref{ss:LDPlev1}, the elements of $\cA(M)$ will play the role
of \textit{potentials}, and hence we shall often refer to them as 
\textit{asymptotically additive potential sequences\/}.

\begin{remark}
An obvious example of an asymptotically additive   potential sequence  is $\GG =
\{S_n G\}$, where $G\in C(M)$. We shall refer to this special case as
\textit{additive}. Some other conditions, which either imply asymptotic  
additivity or
are equivalent to it, are given in Theorem~\ref{prop:characterizationAA}. There,
we prove in particular that if for some $G\in B(M)$ the sequence $\GG = \{S_n
G\}$ has \textit{tempered variation\/} (a condition which is weaker than the 
continuity of $G$), then $\GG$ is asymptotically additive. The tempered
variation condition, which to the best of our knowledge goes back
to~\cite{kessebohmer-2001} (see also~\cite{barreira-2006}), holds in particular
if $G$ satisfies the \textit{bounded variation\/} condition of~\cite{ruelle-1992}.
Another class of examples is given by \textit{weakly almost additive\/}
potentials,\footnote{In the literature, the special case where $C_n=C$ is often
called \textit{almost additive}, and we shall use this convention in the sequel;
see~\cite{Bar2011}.} which are characterized by the following property: there is
a sequence $\{C_n\}_{n\geq 1}\subset\R$ such that $\lim_{n\rightarrow
\infty}n^{-1} C_n=0$ and
\begin{equation}\label{eq:weaklyalmostdec}
-C_m + G_m +G_n\circ\varphi^m \leq G_{m+n}\leq C_m + G_m +G_n\circ\varphi^m, 
\quad m, n\geq 1.
\end{equation}
If a family $\GG\subset C(M)$ is weakly almost additive, then it is
asymptotically additive with $G^{(k)} = k^{-1}G_k$; see Lemma~\ref{p4.30}.
\end{remark}

\begin{remark}\label{rem:equivrel} 
Note that $\cA(M)$ is a vector space on which the seminorm defined by
\begin{equation}\label{eq:nm1Gnfini}
\|\GG\|_*=\limsup_{n\to\infty} n^{-1} \|G_n\|_\infty
\end{equation}
induces the natural equivalence relation: $\GG \sim \GG'$ iff $\|\GG - \GG' \|_*
= 0$ (finiteness of~\eqref{eq:nm1Gnfini} follows immediately
from~\eqref{eq:defasymadd}). As mentioned in~\cite[Remark A.6 (ii)]{FH-2010}
(see also the beginning of Section~3.2 in~\cite{bomfim_multifractal_2015}),
equivalent potential sequences share many important properties and, in
particular, have the same approximating sequences. Furthermore, if $V, V' \in
C(M)$ are such that $V'-V=U-U\circ \varphi$  for some $U\in C(M)$, then $\{S_n
V\} \sim \{S_n V'\}$, so that this concept of equivalence generalizes the
standard notion of equivalence for potentials. Moreover, by the definition of
asymptotic additivity, for all $\GG\in \cA(M)$ we have $\lim_{k\to\infty}\|\GG -
\{S_n G^{(k)}\}\|_* = 0$, so that the additive potential sequences  are dense in
the quotient space $\cA(M)/{\sim}$. Finally, we note that each equivalence
class admits a representative $\GG\subset C(M)$ (see
Remark~\ref{rem:remcharac}).
\end{remark}

\begin{remark}
To the best of our knowledge, the first extension of the classical thermodynamic
formalism of Ruelle and Walters~\cite{ruelle2004,walters1982} beyond the
additive setting goes back to the work of
Falconer~\cite{falconer_subadditive88}. This and later  extensions were
principally motivated by the multifractal analysis of certain classes of
self-similar sets, and in this context the subject has developed rapidly; see
for example~\cite{cao_thermodynamic_2008,FH-2010,zhao_asymptotically_2011,Bar2011,
varandas_weak_2015,iommi_weak_2016} and references therein.

It is likely that the subject will continue to flourish with an expanding number
of applications that cannot be reached within the classical theory; see
Example~\ref{em} below and recent works~\cite{BJPP-2017,BCJPP-2017} for
applications to the theory of repeated quantum measurement processes.
\end{remark}

Theorems~\tA and~\tB extend to asymptotically additive potential
sequences with the following notational changes. Given $\GG=\{G_n\}\in\aA(M)$,
one defines a sequence of probability measures on~$M$ by the relations
(compare with~\eqref{2.8add})
\begin{equation}\label{add}
\IP_n(\dd y)=Z_n^{-1}\sum_{x\in M_n}\eee^{G_n(x)}\delta_x(\dd y), \quad
Z_n=Z_n(\GG)=\sum_{x\in M_n}\eee^{G_n(x)}.
\end{equation}
The time-reversal operation is now defined as $\theta_n=\theta\circ\varphi^{n-1}$. Let us set 
\begin{equation*}
\widehat\IP_n(\dd y)=(\IP_n\circ\theta_n)(\dd y)
=Z_n^{-1}\sum_{x\in M_n}\eee^{G_n\circ\theta\circ \varphi^{n-1}(x)}\delta_x(\dd y)
\end{equation*}
and remark that this relation coincides with~\eqref{1.5} in the case of additive potentials.   We also note that 
\begin{equation*}
\log\frac{\dd\IP_n}{\dd\widehat\IP_n}(x)=
\sigma_n(x) \quad\mbox{for $x\in M_n$,}
\end{equation*}
where we write 
\beq
\sigma_n=G_n-G_n\circ\theta_n
\label{Eq:SigmaDef}
\eeq
for the \textit{entropy production in time~$n$\/}. Accordingly, the ergodic
averages $n^{-1}S_n\sigma$ are now replaced by $n^{-1}\sigma_n$. With the above 
notational changes Theorems~\tA and~\tB hold for any  $\GG\in{\cal
A}(M)$. Starting with Section~\ref{s1b} we shall work exclusively in the
asymptotically additive setting.

\begin{example}[Boundary terms in spin chains]
Consider the spin chains discussed in Example~\ref{e1} in the case where $M$ is
the full two-sided shift (\ie when $A_{ij} = 1$ for all $i,j\in\cA$). 
Since the set $\cA^n$ of configurations of a chain of length $n$ is in one-to-one
correspondence with the set $M_n$ of orbits of period $n$ of $(M,\varphi)$,
the measures $\P_n$ of Theorem~\tA can be interpreted as Gibbs measures on 
finite-size systems. The case $G_n = S_n A_\varPhi$
then corresponds to periodic boundary conditions, while other boundary
conditions lead to asymptotically additive sequences of the kind $G_n = S_n
A_\varPhi + g_n$, where the boundary term $g_n$ satisfies
$\lim_{n\to\infty}n^{-1}\|g_n\|_\infty = 0$.
\end{example}

\begin{example}[Matrix product potentials]\label{em}
Perhaps the best known examples of asymptotically additive potential sequences
arise through matrix products. Denote by ${\mathbb M}_N(\cc)$ the algebra of
all complex $N\times N$ matrices. Let ${\cal M}:M\to{\mathbb M}_N(\cc)$
be a continuous map such that
\[
{\cal M}(x){\cal M}(\varphi(x))\cdots {\cal M}(\varphi^{n-1}(x))\not=0
\]
for all $n\geq 1$ and $x\in M$. The potential sequences of the form
\beq
G_n(x)=\log \|{\cal M}(x){\cal M}(\varphi(x))\cdots
{\cal M}(\varphi^{n-1}(x))\|
\label{rai}
\eeq 
arise in multifractal analysis of self-similar sets, see for
example~\cite{FO-2003,feng-2003,Feng2009,Bar2011}. Sequences of this type also
describe the statistics of some important classes of repeated quantum
measurement processes~\cite{BJPP-2017,BCJPP-2017}. Except in very special cases, the
sequence $\GG=\{G_n\}$ is not additive. Note that the upper almost additivity
\begin{equation}\label{eq:quasibernoulli1}
G_{n+m}\leq C + G_m + G_n\circ \varphi^m
\end{equation} 
always holds with  a constant $C$ depending only on the choice of the matrix
norm on ${\mathbb M}_N(\cc)$. If the entries of ${\cal M}(x)$ are strictly
positive for all $x\in M$, or if $N=2$ and ${\cal M}$ satisfies the cone
condition of~\cite{BarGel06}\footnote{This type of condition can be traced back
to~\cite{Ruelle1979}. See also~\cite[Definition 11.2.1]{Bar2011}} (in the
context of nonconformal repellers), then one can show that
\begin{equation}\label{eq:quasibernoulli2}
-C + G_m + G_n\circ \varphi^m \leq G_{n+m}
\end{equation}  
for some $C>0$, so that $\GG$ is almost additive. In many interesting examples,
however, \eqref{eq:quasibernoulli2} fails, but $\GG$ remains asymptotically
additive and hence our results apply. When the potential defined in~\eqref{rai}
is not asymptotically additive, it can exhibit a very singular behavior from
the thermodynamic formalism point of view.\footnote{In some cases, this
singularity depends on the number-theoretic properties of the entries of
${\cal M}(x)$; see~\cite{BCJPP-2017} for a discussion.} For reasons of space we
postpone the discussion of the last point  to the forthcoming
articles~\cite{CJPS-2017,BCJPP-2017}.
\end{example} 

\begin{example}
Let $(M,\varphi)$ and $\theta$ be as in Example~\ref{e1}, and let $\P$ be any
fully-supported $\varphi$-invariant measure on $M$. For all $n\geq 1$, define 
$G_n\in C(M)$ by 
\[
G_n(x)=\log\P\{y\in M: y_j = x_j\text{ for }j=1,\ldots,n\}.
\] 
Then, if $\GG=\{G_n\}$ is asymptotically additive, the measure
$\P$ is weak Gibbs with respect to $\GG$, and hence Theorem~\tB applies. In
this setup, \eqref{eq:weaklyalmostdec} is interpreted as ``weak dependence.'' In
particular, all invariant quasi-Bernoulli measures (\ie satisfying
\eqref{eq:quasibernoulli1} and \eqref{eq:quasibernoulli2}) on $M$ are weak Gibbs
with respect to an asymptotically additive potential.
\end{example}

We finish with the following general remarks. 

\begin{remark}
To summarize, the contribution of our paper is two-fold. Firstly, to the best of
our knowledge, the POFP has not appeared previously in the literature and
provides a rather general formulation and proof of the FT and FR in the context of
chaotic dynamical systems on compact metric spaces.  Furthermore, the GFP
extends the FT and FR of~\cite{MV-2003} to weak Gibbs measures (which do not
even need to be invariant).  In particular, this extends the validity of the FT
and FR to the phase transition regime, as illustrated in Example~\ref{e1}. Both
results hold for any asymptotically additive potential sequence. Secondly, the
FR for the rate function of the empirical measures (the first relation
in~\eqref{1.8}) is new and  we plan to investigate it further in  other models
of relevance to non-equilibrium statistical mechanics.
\end{remark}

\begin{remark}
To the best of our knowledge, the FT and FR in the phase transition regime have
not been previously discussed in the physics and mathematics literature, apart
from stochastic lattice gases; see~\cite{BDGJL-2006,BDGJL-2015}.  On the other
hand, a considerable amount of effort in the dynamical systems community over
the last two decades has been devoted to the extension of multifractal analysis
to the phase transition regime; for instance,
see~\cite{makarov-1998,FO-2003,testud-2006}.  Given the link between
multifractal analysis and  large deviations
theory~\cite{DK-1999,kessebohmer-2001}, the two research directions are related,
and this connection remains to be investigated in the future.
\end{remark}

\begin{remark}
Although the  conceptual emphasis of this paper has been on the FT and  FR for
entropy production generated by a  reversal operation, the LDP parts of
Theorems~\tA and~\tB are  of independent interest and have  wider
applicability; see~\cite{CJPS-2017}  for a more general approach than the one
adopted here,  and~\cite{BCJPP-2017} for some concrete applications in the
context of repeated quantum measurement processes.
\end{remark}

The paper is organized as follows. Section~\ref{s1} is a continuation of the
introduction where  we review the general mathematical structure and
interpretation of the FR and FT from a modern point of view. In
Section~\ref{s1b} we collect preliminaries needed for the formulations and
proofs of our results, including an overview of the asymptotically additive
thermodynamic formalism. Section~\ref{s2} is devoted to the POFP and
Section~\ref{s4} to the GFP.  Our main technical results regarding the LDP are
stated and proven in Section~\ref{s:abstractLevel3}. Finally, in
Section~\ref{sec:charac} we discuss some properties and characterizations of
asymptotically  additive potential sequences.

This work is accompanied by a review article~\cite{CJPS_phys} where the reader
can find additional information and examples regarding the FT and FR.

\bigskip
\textbf{Acknowledgments.}  We are grateful to R.~Chetrite and H.~Comman for providing us with useful references, and to L.~Bruneau for helpful comments. This research was supported by
the \textit{Agence Nationale de la Recherche\/} through the grant NONSTOPS
(ANR-17-CE40-0006-01, ANR-17-CE40-0006-02, ANR-17-CE40-0006-03), the CNRS
collaboration grant \textit{Fluctuation theorems in stochastic systems\/}, and the
\textit{Initiative d'excellence Paris-Seine\/}. NC was supported by Swiss National
Science Foundation Grant 165057. VJ acknowledges the support of NSERC. The work
of CAP has been carried out in the framework of the Labex Archim\`ede
(ANR-11-LABX-0033) and of the A*MIDEX project (ANR-11-IDEX-0001-02), funded by
the ``Investissements d'Avenir'' French Government programme managed by the
French National Research Agency (ANR). The research of AS was carried out within
the MME-DII Center of Excellence (ANR-11-LABX-0023-01).

\section{Prologue: what is the Fluctuation Theorem?}
\label{s1}

\subsection{Transient fluctuation relations}
\label{s1.1}

Our starting point is a family of probability spaces $(\Omega_n,\FF_n,\IP_n)$
indexed by a parameter~$n\in {\mathbb N}$. Each of these spaces is equipped with
a measurable involution $\Theta_n:\Omega_n\to\Omega_n$ called \textit{reversal\/}
(the map~$\Theta_n$ is its own inverse).  In many cases of interest the
probability space $(\Omega_n,\FF_n,\IP_n)$ describes the space-time statistics
of the physical system under consideration over the finite time
interval~$[0,n]$, and the map $\Theta_n$ is related to time-reversal.

Let us set $\widehat\IP_n=\IP_n\circ\Theta_n$ and impose the 
following hypothesis:
\begin{description}
\item[(R) Regularity.] 
\textit{The measures $\widehat\IP_n$ and $\IP_n$ are equivalent.} 
\end{description}
Under Assumption~\textbf{(R)}, one defines 
\begin{equation*}
\sigma_n=\log\frac{\dd\IP_n}{\dd\widehat\IP_n}. 
\end{equation*}
This is a real-valued random variable on~$\Omega_n$, and we denote
by~$P_n$ its law under~$\IP_n$. The very definition
of~$\sigma_n$ implies a number of simple, yet important properties.

\smallskip
\textbf{Relative entropy.} The relative entropy\footnote{Also called 
Kullback--Leibler divergence.} of~$\IP_n$ with respect 
to~$\widehat\IP_n$ is given by the relation
\[
\Ent(\IP_n\,|\,\widehat\IP_n) = \int_{\Omega_n}
\log\left(\frac{\dd\IP_n}{\dd\widehat\IP_n}\right)\dd\IP_n
=\int_{\Omega_n}\sigma_n\dd\IP_n
=\int_\R s\, P_n(\dd s). 
\]
Since this quantity is non-negative, we obtain 
\begin{equation} \label{4.2}
\int_\R s\,P_n(\dd s)\ge0,
\end{equation}
which asserts that under the law~$\IP_n$, positive values 
of~$\sigma_n$ are favored.

\smallskip
\textbf{R\'enyi entropy}. R\'enyi's relative $\alpha$-entropy 
of~$\widehat\IP_n$ with respect to~$\IP_n$ is defined by 
\[
\Ent_\alpha(\IP_n\,|\,\widehat\IP_n)
=\log\int_{\Omega_n}
\left(\frac{\dd\widehat{\IP}_n}{\dd\IP_n}\right)^\alpha
\dd\IP_n
=\log\int_{\Omega_n}\eee^{-\alpha\sigma_n}\dd\IP_n
=\log \int_\R\eee^{-\alpha s}P_n(\dd s)=: e_n(\alpha).
\]
The function $\R\ni\alpha\mapsto e_n(\alpha)\in]-\infty,+\infty]$
is convex and lower semicontinuous. It vanishes at $\alpha=0$ and $\alpha=1$, so 
that~$e_n(\alpha)$ is non-positive and finite on~$[0,1]$, and
non-negative outside~$[0,1]$. It admits an analytic continuation to the strip
$\{z\in{\mathbb C}:0<\Re z<1\}$ which is continuous on its closure. Expressing
the relation $e_n(1)=0$ in terms of~$P_n$, we derive 
\begin{equation*}
\int_\R\eee^{-s}P_n(\dd s)=1.
\end{equation*}
In the physics literature this relation is sometimes called \textit{Jarzynski's 
identity}.

\begin{proposition} \label{p4.1}
In the above setting, the following two relations hold:
\begin{gather}
e_n(\alpha)=e_n(1-\alpha)\quad\text{for }\alpha\in\R,
\label{4.4}\\[6pt]
\frac{\dd P_n}{\dd \widehat P_n}(s)=\eee^s
\quad\text{for }s\in\R,
\label{4.5}
\end{gather}
where $\widehat P_n$ is the image of~$P_n$ under the reflection 
$\vartheta(s)=-s$. 
\end{proposition}
\begin{remark}
Relations~\eqref{4.4} and~\eqref{4.5} are in fact equivalent: the validity of
one of them implies the other. We shall refer to them  as the \textit{transient
FR\/}. It implies and refines~\eqref{4.2},
and its basic appeal is its universal form. In
applications to non-equilibrium physics, the transient FR is a fingerprint of
time-reversal symmetry breaking and emergence of the 2$^\mathrm{nd}$ Law of
Thermodynamics.
\end{remark}
\begin{proof}{ of Proposition~\ref{p4.1}}
Relation~\eqref{4.4} is a simple consequence of a symmetry property of 
R\'enyi's entropy:
\begin{align*}
e_n(1-\alpha)&=\Ent_{1-\alpha}(\IP_n\,|\,\widehat\IP_n)
=\Ent_\alpha(\widehat\IP_n\,|\,\IP_n)
=\log\int_{\Omega_n}\eee^{\alpha\sigma_n}\dd\widehat\IP_n\\
&=\log\int_{\Omega_n}\eee^{\alpha\sigma_n\circ\Theta_n}
\dd(\widehat\IP_n\circ\Theta_n)
=\log\int_{\Omega_n}\eee^{-\alpha\sigma_n}\dd\IP_n
=e_n(\alpha),
\end{align*}
where we used the elementary relation $\sigma_n\circ\Theta_n=-\sigma_n$. 

To prove~\eqref{4.5}, we exponentiate~\eqref{4.4} and rewrite the result in terms 
of~$P_n$:
\[
\int_\R \eee^{-\alpha s}P_n(\dd s)
=\int_\R \eee^{-(1-\alpha)s}P_n(\dd s)
=\int_\R \eee^{-\alpha s}(\eee^s\widehat P_n)(\dd s). 
\]
Using now the analyticity of the function~$e_n(z)$ in the open strip $0<\Re z<1$, 
its continuity in the closed strip $0\leq \Re z \leq 1$, and the
fact that the characteristic function uniquely defines the corresponding
measure, we deduce that~$P_n(\dd s)$ 
and~$(\eee^s\widehat P_n)(\dd s)$ coincide. This is equivalent 
to~\eqref{4.5}. 
\hfill\qed
\end{proof}

\subsection{Fluctuation Theorem and Fluctuation Relation}
\label{s1.2}

\begin{definition} \label{d1.2}
We shall say that the \textup{Fluctuation Theorem} holds for the family 
$(\Omega_n,\FF_n,\IP_n,\Theta_n)$ if there is a lower semicontinuous  function 
$I:\rr \to[0,+\infty]$  such that,  for any Borel set $\Gamma\subset\rr$,
\begin{equation} \label{4.6}
\begin{split}
-\inf_{s\in \bdot\Gamma} I(s)
&\le \liminf_{n\to\infty}n^{-1}
\log \IP_n\{n^{-1}\sigma_n\in\Gamma\}\\[2mm]
&\le \limsup_{n\to\infty}n^{-1}
\log \IP_n\{n^{-1}\sigma_n\in\Gamma\}
\le -\inf_{s\in \overline\Gamma}I(s), 
\end{split}
\end{equation}
where $\bdot\Gamma/\overline\Gamma$ denotes the interior/closure of $\Gamma$. 
\end{definition}
Let us note that $\IP_n\{n^{-1}\sigma_n\in\Gamma\}=P_n(n\Gamma)$,
so that~\eqref{4.6} can be rewritten as 
\begin{equation} \label{4.8}
-\inf_{s\in \bdot\Gamma} I(s)
\le \liminf_{n\to\infty}n^{-1}\log P_n(n\Gamma)
\le \limsup_{n\to\infty}n^{-1}\log P_n(n\Gamma)
\le -\inf_{s\in \overline\Gamma}I(s).
\end{equation}
The following result shows that the transient FR implies a symmetry relation for 
the rate function in the FT.
\begin{proposition} \label{p4.3}
Suppose that the FT holds for a family 
$(\Omega_n,\FF_n,\IP_n,\Theta_n)$. Then 
the corresponding rate function~$I$ satisfies the relation
\begin{equation} \label{4.7}
I(-s)=I(s)+s\quad \hbox{for $s\in \rr$}.
\end{equation}
\end{proposition}
\begin{proof}{} 
In view of~\eqref{4.5}, for any Borel set $\Gamma\subset\R$ we have 
\[
P_n(\Gamma)\le\eee^{\sup \Gamma}P_n(-\Gamma). 
\]
Replacing~$\Gamma$ with~$n\Gamma$ and using~\eqref{4.8} we see that 
\[
-\inf_{s\in\bdot{\Gamma}}I(s)\le 
\liminf_{n\to\infty}n^{-1}\log P_n(n\Gamma)
\le \limsup_{n\to\infty}
n^{-1}\log\bigl(\eee^{n \sup \Gamma}P_n(-n\Gamma)\bigr)
\le \sup \Gamma-\inf_{s\in\overline{\Gamma}}I(-s). 
\]
Taking $\Gamma=]a-\e,a+\e[$ with $\e>0$, we derive
\begin{equation} \label{4.9}
\inf_{|s+a|<2\e}I(s)\le\inf_{|s+a|\le\e}I(s)
\le a+\e+\inf_{|s-a|<\e}I(s).
\end{equation}
Since the function~$I$ is lower semicontinuous, we have 
$I(a)=\lim_{\e\downarrow0}\inf_{|s-a|<\e}I(s)$. Passing to the limit 
in~\eqref{4.9} as $\e\downarrow 0$, we obtain
\[
I(-a)\le a+I(a) 
\]
for any $a\in\rr$. Replacing~$a$ by~$-a$ and comparing the two inequalities, we 
arrive at~\eqref{4.7}. 
\hfill\qed 
\end{proof}

\subsection{Entropic pressure}
\label{s1.4}

Suppose that the Fluctuation Theorem (Definition~\ref{d1.2}) holds. Then under 
very general conditions Varadhan's lemma~\cite[Theorem~4.3.1]{DZ2000} implies 
that the limit 
\begin{equation} \label{4.13}
e(\alpha)=\lim_{n\to\infty}n^{-1}e_n(\alpha)
\end{equation}
exists for all $\alpha\in\rr$ and that 
\begin{equation}\label{eq:ealphalegendre}
e(\alpha)=-\inf_{s\in\R}\bigl(s\alpha+I(s)\bigr)
=\sup_{s\in\R}\bigl(-s\alpha-I(s)\bigr).	
\end{equation}
The function~$e$ is called the \textit{entropic pressure\/} of the family 
$(\Omega_n,\FF_n,\IP_n,\Theta_n)$. 
Elementary arguments show that:
\begin{itemize}
\item[\textbf{(a)}] 
$\rr\ni\alpha\mapsto e(\alpha)\in]-\infty,+\infty]$ is a closed proper convex 
function;\footnote{\ie it is convex, lower semicontinuous and not everywhere 
infinite.}
\item[\textbf{(b)}]  
it is  non-positive on~$[0,1]$ and non-negative outside~$[0,1]$, with a global 
minimum at $\alpha=1/2$;
\item[\textbf{(c)}] 
it satisfies the relations $e(0)=e(1)=0$ and 
\begin{equation*}
e(\alpha)=e(1-\alpha)\quad\mbox{for all $\alpha\in\R$}.
\end{equation*}
\end{itemize}
If the rate function~$I$ is convex on $\rr$, then inverting the Legendre 
transform~\eqref{eq:ealphalegendre} gives
\begin{equation} \label{4.15}
I(s)=-\inf_{\alpha\in\R}\bigl(s\alpha +e(\alpha)\bigr)
=\sup_{\alpha\in\R}\bigl(s\alpha-e(-\alpha)\bigr).
\end{equation} 
If $I$ is not convex on~$\rr$, the same relation holds if~$I$ is replaced by its 
lower convex envelope. 

The above discussion can be turned around. Suppose that limit~\eqref{4.13}
exists and that the entropic pressure $e$ is differentiable on~$\rr$. Then the
G\"artner--Ellis  theorem~\cite[Section~2.3]{DZ2000} implies that the
Fluctuation Theorem holds with the convex lower semicontinuous rate function $I$
 given by~\eqref{4.15}. This gives a technical route to prove the Fluctuation
Theorem. Since the seminal work~\cite{LS-1999}, this route has been dominant in
mathematical approaches to the FT.

Returning to the dynamical system~$(M, \varphi)$, the above route yields a quick
proof of Theorem~\tA if  the potentials~$G$  and $G\circ \theta$ are 
Bowen-regular. We follow~\cite{JPRB-2011}. By a classical result of Bowen, the
limit
\[
e(\alpha)=\lim_{n\rightarrow \infty}\frac{1}{n}\log 
\int_{M_n}\eee^{-\alpha S_n \sigma}\d \P_n
\]
exists for all $\alpha\in\R$ and is equal to the topological pressure of the 
potential $(1-\alpha) G + \alpha G\circ \theta$; see also Theorem~\ref{Thm:USCh} 
in Section~\ref{s2.4}. By another classical result of Bowen~\cite{bowen1975}, the
potential $(1-\alpha)G+\alpha G\circ\theta$ has unique equilibrium state for
all $\alpha$, and~\cite[Theorem 9.15]{walters1982} shows that $e$
is differentiable on~$\rr$. Thus, the  G\"artner--Ellis theorem applies and 
gives the FT and the second FR in~\eqref{1.8}. Assuming that the vector space of
all Bowen-regular potentials in dense in~$C(M)$ (for instance, this is the case
in Examples~\ref{e1}--\ref{e3} of the introduction), Kifer's
theorem~\cite{kifer-1990} implies that the LDP part of Theorem~\tA holds.
In this case the first FR in~\eqref{1.8} follows from a computation given in
Section~\ref{s2.3}. The same proof applies verbatim to Theorem~\tB, and in 
particular recovers the results of~\cite{MV-2003}.

The novelty of Theorems~\tA and~\tB is that they hold for potentials~$G$
for which the entropic pressure is not necessarily differentiable, hence in the
phase transition regime. In this case the proof follows a different strategy:
one first proves the LDP for empirical measures, and then uses the 
contraction principle to prove the FT. Another novelty is that these results also 
extend to asymptotically additive potential sequences.

\subsection{What does the Fluctuation Theorem mean?}
\label{s1.3}

Returning to the level of generality of Sections~\ref{s1.1} and~\ref{s1.2}, the
interpretation of the rate function~$I$ in the FT is given in terms of
\textit{hypothesis testing error exponents\/} of the family $\{(\IP_n,\widehat\IP_n)\}$.
These exponents describe the rate of separation between the measures~$\IP_n$
and~$\widehat\IP_n$ as~$n\to\infty$. If the elements of~${\mathbb N}$ are
instances of time and~$\Theta_n$ is related to time-reversal, these exponents
quantify the emergence of the arrow of time and can be viewed as a fine form of
the second law of thermodynamics. Let us recall the definition of the three
types of error exponents relevant to our study and state some results without
proofs, which can be found in~\cite{JOPS-2012,JNPPS}. 
All the dynamical systems discussed in this paper give rise to a convex rate function
in the FT, and we shall  restrict our attention to this case in Propositions~\ref{prop:chernoff} and \ref{prop:hoeffding} below.

\smallskip
\textbf{Stein error exponents}. {\itshape Given $\gamma\in]0,1[$, we set 
\begin{equation*}
\ssss_\gamma(\IP_n,\widehat\IP_n)
=\inf\bigl\{\widehat\IP_n(\Gamma):
\Gamma\in\FF_n,\IP_n(\Gamma^c)\le\gamma\bigr\}. 
\end{equation*}
The lower and upper Stein exponents of the family $\{(\IP_n,\widehat{\IP}_n)\}$ 
are defined by
\begin{align*}
\underbar{$\ssss$}&=\inf\Bigl\{\liminf_{n \rightarrow \infty}n^{-1}\log\widehat\IP_n(\Gamma_n):
\Gamma_n\in\FF_n,\lim_{n \rightarrow \infty }\IP_n(\Gamma_n^c)=0\Bigr\},\\
\bar\ssss
&=\inf\Bigl\{\limsup_{n \rightarrow \infty }n^{-1}\log\widehat\IP_n(\Gamma_n):
\Gamma_n\in\FF_n,\lim_{n \rightarrow \infty }\IP_n(\Gamma_n^c)=0\Bigr\}.
\end{align*}
}%
The following result establishes a link between the large-$n$ asymptotics
of~$\ssss_\gamma(\IP_n,\widehat\IP_n)$, the Stein exponents, the weak law of 
large numbers and the FT.
\begin{proposition}
\ben
\item Suppose that $n^{-1}\sigma_n$ converges in probability to a deterministic
limit~$\ep$. Then, for any $\gamma\in]0,1[$,
\begin{equation*}
\lim_{n\to\infty}n^{-1}
\log \ssss_\gamma(\IP_n,\widehat\IP_n)=\underbar{$\ssss$}=\bar\ssss=-\ep. 
\end{equation*}
\item If the FT holds and the corresponding rate function~$I$ vanishes at a 
unique point $\underline{s}\in\R$, then $n^{-1}\sigma_n$ converges in probability 
to $\underline{s}$ and hence Part~(1) holds with $\ep=\underline{s}$.
\item If, in addition, the entropic pressure defined by~\eqref{eq:ealphalegendre} is  
differentiable at $\alpha=0$, then 
\[
\lim_{n \rightarrow \infty}n^{-1}\Ent(\P_n|\wP_n)=\ep=-e^\prime(0).
\]
\een
\end{proposition}
Recall that, for two probability measures~$\IP$ and~$\Q$ on a measurable 
space~$(\Omega,\FF)$, the \textit{total variation distance\/}  is given by 
\begin{equation*} 
\|\IP-\Q\|_{\mathrm{var}}
=\sup_{\Gamma\in\FF}|\IP(\Gamma)-\Q(\Gamma)|
=1-\int_\Omega (\Delta\wedge 1)\dd\Q,
\end{equation*}
where the second equality holds if $\IP$ is absolutely continuous with respect 
to~$\Q$, and~$\Delta$ stands for the corresponding density. Note that
$\|\IP-\Q\|_{\mathrm{var}}\le1$ with equality iff $\IP$ and $\Q$ are mutually
singular. The Chernoff exponents of the family $\{(\IP_n,\widehat{\IP}_n)\}$
quantify the exponential rate in the convergence
$\|\IP_n-\widehat{\IP}_n\|_{\mathrm{var}}\to1$.

\smallskip
\textbf{Chernoff error exponents.} {\itshape The lower and upper Chernoff exponents are 
defined by} 
\begin{equation*}
\underline{c}=\liminf_{n\to\infty}n^{-1}
\log\bigl(1-\|\IP_n-\widehat\IP_n\|_{\mathrm{var}}\bigr),\quad
\bar{c}=\limsup_{n\to\infty}n^{-1}
\log\bigl(1-\|\IP_n-\widehat\IP_n\|_{\mathrm{var}}\bigr). 
\end{equation*}%
The following result provides a connection between these exponents and the 
FT.
\begin{proposition}\label{prop:chernoff}
Suppose that the FT holds with a convex rate function~$I$. Then the upper and 
lower Chernoff exponents coincide, and their common value is given by\footnote{Since $I$ is assumed to be convex, one has 
$-I(0)=e(1/2)=\min\{e(\alpha):\alpha\in\rr\}$.} 
\begin{equation*}
\underline{c}=\bar{c}=
\lim_{n\to\infty}n^{-1}
\log\bigl(1-\|\IP_n-\widehat\IP_n\|_{\mathrm{var}}\bigr)=-I(0).
\end{equation*}
\end{proposition}
Thus, if $I(0)>0$, then the measures~$\IP_n$ and~$\widehat\IP_n$ concentrate on
the complementary subsets $\{\sigma_n>0\}$ and~$\{\sigma_n<0\}$, respectively,
and separate with an exponential rate~$-I(0)$:
\[
\lim_{n\to\infty}n^{-1}\log\IP_n\{\sigma_n<0\}
=\lim_{n\to\infty}n^{-1}\log\widehat\IP_n\{\sigma_n>0\}
=-I(0). 
\]

\smallskip
\textbf{Hoeffding error exponents}.
{\itshape Given $r\in\R$, the Hoeffding exponents are defined by  
\begin{align*}
\underline{\hhhh}(r)
&=\inf\Bigl\{\liminf_{n \rightarrow 
\infty}n^{-1}\log\widehat\IP_n(\Gamma_n):\Gamma_n\in\FF_n,
\limsup_{n \rightarrow \infty}n^{-1}\log\IP_n(\Gamma_n^c)<-r\Bigr\},\\ 
\bar\hhhh(r)
&=\inf\Bigl\{\limsup_{n \rightarrow 
\infty}n^{-1}\log\widehat\IP_n(\Gamma_n):\Gamma_n\in\FF_n,
\limsup_{n \rightarrow \infty}n^{-1}\log\IP_n(\Gamma_n^c)<-r\Bigr\},\\
\hhhh(r)
&=\inf\Bigl\{\lim_{n \rightarrow 
\infty}n^{-1}\log\widehat\IP_n(\Gamma_n):\Gamma_n\in\FF_n,
\limsup_{n \rightarrow \infty}n^{-1}\log\IP_n(\Gamma_n^c)<-r\Bigr\},
\end{align*}
where the infimum in the last relation is taken over all families~$\{\Gamma_n\}$ 
for which the limit exists.}

The following proposition gives some elementary properties of these exponents and
their connections with the FT.
\begin{proposition}\label{prop:hoeffding}
\ben 
\item The functions~$\underline{\hhhh}$, $\bar\hhhh$, and~$\hhhh$ take the 
value~$-\infty$ on~$]-\infty,0[$ and are non-decreasing and non-positive 
on~$[0,+\infty[$. Moreover, the inequalities
$\underline{\hhhh}(r)\le\bar\hhhh(r)\le\hhhh(r)\le 0$ hold for $r\ge0$.
\item Suppose that the FT holds with a convex rate function~$I$. Then
the Hoeffding exponents coincide, and their common value is given in terms of the 
entropic pressure~\eqref{eq:ealphalegendre} by
\begin{equation}\label{eq:deffinfar} 
\rr\ni r\mapsto\underline{\hhhh}(r)=\bar\hhhh(r)=\hhhh(r)
=-f(r):=\inf_{\alpha\in[0,1[}\frac{\alpha r+e(\alpha)}{1-\alpha}.
\end{equation} 
\item Let $\underline{s}=\inf\{s\in\rr:I(s)=0\}$, then $\underline{s}=\partial^-e(1)$.%
\footnote{$\partial^-$ denotes the left derivative.} The rate function $I$ is strictly decreasing on the interval $[-\underline{s},\underline{s}]$ and maps this
interval onto $[0,\underline{s}]$. We denote by 
$I^{-1}:[0,\underline{s}]\to[-\underline{s},\underline{s}]$ the inverse function.
\item The function $f$ defined in \eqref{eq:deffinfar} is a closed proper convex function on $\rr$ such that $f(0)=\underline{s}$
and $f(r)=0$ for all $r\ge\underline{s}$.
Moreover, the restriction of $f$ to the interval $[0,\underline{s}]$
is an involutive homeomorphism satisfying
\[
2I(0)-r\le f(r)=r+I^{-1}(r)\le\underline{s}-r,
\]
and in particular $f(I(0))=I(0)$.
\een
\end{proposition}

\subsection{Interpretation of Theorems A and B}
\label{s1.5}

The FT part of Theorem~\tA fits directly into the mathematical framework and
interpretation of the FT discussed in Section~\ref{s1.3} with $\Omega_n=M_n$,
$\P_n$ defined by~\eqref{2.8add}, and $\Theta_n=\theta \circ \varphi^{n-1}$.

The above interpretation does not apply  to the setup of Theorem~\tB. The FT
part of Theorem~\tB is related to the \textit{principle of  regular entropic
fluctuations\/} of~\cite{JPRB-2011} adapted to the setting of this paper, and
provides a \textit{uniformity\/} counterpart to the FT of Theorem~\tA which is
both of conceptual and practical (numerical, experimental) importance. Here we
shall briefly comment on this point, referring the reader 
to~\cite{CJPS_phys,JNPPS} for additional discussion.

Let $\P_n$ be as in  Theorem~\tA. For each fixed $m\geq 0$, since the
function $\sigma$ is bounded, the random variables~$(n+m)^{-1}S_{n+m}\sigma$
and~$n^{-1}S_n\sigma$ are exponentially equivalent under the law~$\IP_{n+m}$ as
$n\to\infty$, and the  theorem implies that for any Borel set $\Gamma \subset
\rr$,
\beq
\begin{split}
-\inf_{s\in \bdot\Gamma} I(s)
&\le \liminf_{n\to\infty}n^{-1}
\log \IP_{n+m}\left\{n^{-1}S_n\sigma\in\Gamma\right\}\\[2mm]
&\le \limsup_{n\to\infty}n^{-1}
\log \IP_{n+m}\left\{n^{-1}S_n\sigma \in\Gamma\right\}
\le -\inf_{s\in \overline\Gamma}I(s).
\end{split}
\label{wed-ti}
\eeq
If, along some subsequence, $\P_{m_k}\rightharpoonup\P$, a natural question is
whether~\eqref{wed-ti} holds with $\P_{n+m}$ replaced by $\P$. Theorem~\tB
gives a positive answer if $\P$ is a weak Gibbs measure  with the same potential
$G$ as in the sequence $\P_n$, and further asserts  that this is the only
requirement for $\P$; neither weak convergence nor invariance of $\P$ play a
role. Such level of uniformity is a somewhat surprising strengthening of the
principle of regular entropic fluctuations in the setting  of chaotic dynamical
systems on compact metric spaces.

\subsection{Outlook}
The general mathematical framework and interpretation of the FT presented in
this section are rooted in pioneering works on the
subject~\cite{ECM-1993,ES-1994,GC-1995,GC-1995a,LS-1999,maes-1999}. As
formulated here, they were developed in special cases
in~\cite{JOPS-2012,BJPP-2017}, and will be studied in full generality in~\cite{JNPPS}. When
a richer mathematical/physical structure is available, one can say more. For
example, for open stochastic or Hamiltonian systems which carry energy fluxes
generated by temperature differentials, the entropy production observable
coincides with thermodynamic entropy production;
see~\cite{JPRB-2011,JPS-2017,CJPS_phys}. For the chaotic dynamical systems
$(M,\varphi)$ considered in this paper, Theorems~\tA and~\tB show that
the FT is a structural feature of the thermodynamic formalism. Relaxing the
chaoticity assumptions (expansiveness and specification)  brings forward a
number of important open problems that remain to be discussed in the future.

\section{Preliminaries}
\label{s1b}

\subsection{A class of continuous dynamical systems}
\label{s2.1}

Let $M$ be a compact metric space with metric~$d$ and Borel
$\sigma$-algebra~$\BB(M)$. We recall that~$C(M)$ (respectively, $B(M)$) denotes
the Banach space of real-valued continuous (bounded measurable) functions~$f$ on
$M$ with the norm $\|f\|_\infty=\sup_{x\in M}|f(x)|$. The set $\PP(M)$ of Borel
probability measures on~$M$ is endowed with the topology of weak convergence
(denoted $\rightharpoonup$) and the corresponding Borel $\sigma$-algebra. Given
$V\in C(M)$ and $\Q\in\PP(M)$, we denote by $\langle V,\Q\rangle$ the integral
of~$V$ with respect to~$\Q$. In the following, we shall always assume:
\newcommand{\hC}{\hyperlink{hyp.C}{\textbf{(C)}}\xspace}
\newcommand{\hH}{\hyperlink{hyp.H}{\textbf{(H)}}\xspace}
\begin{description}
\item[\hypertarget{hyp.C}{\textbf{(C)}}]\label{C}
$\varphi:M\to M$ \textit{is a continuous map.}
\end{description}
On occasions, we shall strengthen the above standing assumption to 
\begin{description}   
\item[\hypertarget{hyp.H}{\textbf{(H)}}]\label{H}
$\varphi:M\to M$ \textit{is a homeomorphism.}
\end{description}
In the sequel we shall always explicitly mention when Condition~\hH is
assumed. The reversal operation as defined in the introduction makes sense only
if~\hH holds. The transformations that lead to the  FR and FT for general
continuous maps are defined in Section~\ref{s2.3}.

The set of $\varphi$-invariant elements of $\PP(M)$ is denoted
by~$\PP_\varphi(M)$, and the set of $\varphi$-ergodic measures  by 
$\EE_\varphi(M)$. The topological entropy of $\varphi$  is denoted by $h_\mathrm{
Top}(\varphi)$. The  Kolmogorov--Sinai entropy of $\varphi$ with respect to
$\Q\in\PP_\varphi(M)$ is denoted by  $h_\varphi(\Q)$.

The \textit{orbit\/} of a point $x\in M$ is defined as
$\{\varphi^k(x)\}_{k\in\zz_+}$. An orbit is said to be $n$\textit{-periodic\/} if
$\varphi^n(x)=x$, and we denote by~$M_n$ the set of fixed points of~$\varphi^n$.

For $I=\llbracket l,m\rrbracket\subset\zz_+$, we denote orbit segments by
\[
\varphi^I(x)=\{\varphi^k(x)\}_{k\in I},
\]
and call \textit{specification} a finite family of such segments
\[
\xi=\{\varphi^{I_i}(x_i)\}_{i\in\llbracket 1,n\rrbracket}.
\]
The integers $n$ and 
$L(\xi)=\max\{|k-k'|:k,k'\in\cup_{i\in\llbracket 1,n\rrbracket}I_i\}$
are called the \textit{rank} and the \textit{length} of $\xi$ respectively.
The specification $\xi$ is $N$\textit{-separated} whenever 
$d(I_i,I_j)=\min_{k_i\in I_i,k_j\in I_j}|k_i-k_j|\ge N$
for all distinct $i,j\in\llbracket 1,n\rrbracket$. It is
$\epsilon$\textit{-shadowed} by $x\in M$ whenever
\[
\max_{i\in\llbracket 1,n\rrbracket}
\max_{k\in I_i}d(\varphi^k(x),\varphi^k(x_i))<\e.
\]
Given $x\in M$, $n\ge0$, and~$\e>0$, the \textit{Bowen ball\/} is defined by
\begin{equation*}
	B_n(x,\e)=\{y\in M:d(\varphi^k(y),\varphi^k(x))<\e\,\mbox{ for }\,0\le k\le n-1\}. 
\end{equation*}

Many variants of the specification property appear in the literature; 
see~\cite{KLO-2016} for a review. We shall make use of the following two forms:
\newcommand{\hWPS}{\hyperlink{hyp.WPS}{\textbf{(WPS)}}\xspace}
\newcommand{\hS}{\hyperlink{hyp.S}{\textbf{(S)}}\xspace}
\begin{description}
\item[\hypertarget{hyp.WPS}{\textbf{(WPS)}}]\label{Cond:WPS}
{\itshape $\varphi$ has the\/ \textup{weak periodic specification property}
if for any $\delta>0$ there is  a sequence of
integers~$\{m_\delta(n)\}_{n\ge n_0}$ such that $0\le m_\delta(n)<n$ for 
$n\ge n_0$, $\lim_{\delta \downarrow 0} \lim_{n\to\infty}n^{-1}m_\delta(n) = 0$, 
and for any $x\in M$ and $n\ge n_0$, we have 
$M_n\cap B_{n-m_\delta(n)}(x,\delta)\ne\varnothing$.}
\item[\hypertarget{hyp.S}{\textbf{(S)}}]\label{Cond:S}
{\itshape $\varphi$ has the\/ \textup{specification property} if
for any $\delta>0$ there is $N(\delta)\ge1$ such that any
$N(\delta)$-separated specification 
$\xi=\{\varphi^{I_i}(x_i)\}_{i\in\llbracket 1,n\rrbracket}$ is $\delta$-shadowed
by some $x\in M$.}
\end{description}

\begin{remark}\label{rBSP}
We shall also refer to \textit{Bowen's specification property\/}~\cite{bowen-1974}
as the property~\hS with the additional constraint that 
$x\in M_{L(\xi)+N(\delta)}$. Bowen's specification obviously implies 
both~\hWPS and~\hS.
\end{remark}

The weak periodic specification property is well suited for the LDP of 
Theorem~\tA. Together with expansiveness (see Definition~\ref{d2.1} below),
it is also sufficient to justify a large part of the thermodynamic formalism
involved in the proof of this result. It is not needed for Theorem~\tB.

The specification Property~\hS is involved in the proof of both
Theorems~\tA and~\tB. However, it is only needed to ensure the
conclusion of the following proposition (see Theorem B in~\cite{EKW-1994}, whose
proof given for case~\hH extends without change to any continuous
$\varphi$). If the conclusion of the latter can be obtained by a different
argument (see Remark~\ref{r4}), then Property~\hS is not needed at all.

\bep\label{Prop:EntropyDensity}
Suppose that $\varphi$ satisfies Condition~\hS.
Then, the set $\EE_\varphi(M)$ of ergodic measures is\/ \textup{entropy-dense} in $\cP_\varphi(M)$, i.e., for any $\IP\in\cP_\varphi(M)$ there exists a 
sequence $\{\IP_m\}\subset\EE_\varphi(M)$ such that
\[
\IP_m\rightharpoonup\IP,\qquad h_\varphi(\IP_m)\to h_\varphi(\IP)
\quad\mbox{as $m\to\infty$}. 
\]
\eep

\subsection{Asymptotic additivity}
In this section we establish three technical results about asymptotically additive
potential sequences. The first  result is related to~\cite[Proposition~A.1~(1)
and Lemma~A.4]{FH-2010}. 
See also~\cite[Proposition~3.2]{bomfim_multifractal_2015}.

\begin{lemma}\label{lem:limiteGGPP}
For all $\IP\in\PP_\varphi(M)$ and~$\GG\in\cA(M)$, the limit
\begin{equation} \label{4.04}
\GG(\IP)=\lim_{n\to\infty}n^{-1}\langle G_n,\IP\rangle
\end{equation}
exists and  is finite. Moreover, for any approximating sequence $\{G^{(k)}\}$ of $\GG$ we have  
\begin{equation} \label{4.04b}
\GG(\IP) = \lim_{k\to\infty} \langle G^{(k)},\IP\rangle.
\end{equation}
The convergence in \eqref{4.04} and \eqref{4.04b} is uniform in $\P\in \cP_\varphi(M)$, and the real-valued function $\IP\mapsto \GG(\IP)$ is continuous on the space~$\PP_\varphi(M)$ endowed with the weak topology.
\end{lemma}

\begin{proof}{}
For all $k,m,n$ and all $\IP\in\PP_\varphi(M)$ we have
\begin{equation}\label{eq:approxCauchy}
\begin{split}
\left|\frac 1 n \langle G_n, \P\rangle - \frac 1 m \langle G_m, \P\rangle\right| 
&\leq 	\left|\frac 1 n \langle G_n, \P\rangle -  \langle G^{(k)}, 
\P\rangle\right| + \left|\frac 1 m \langle G_m, \P\rangle -  \langle G^{(k)}, 
\P\rangle\right|\\
& \leq \frac 1 n  \left\|  G_n - S_n G^{(k)} \right\|_\infty + \frac 1 m  \left\|  G_m - S_m G^{(k)} \right\|_\infty.
\end{split}
\end{equation}

In view of~\eqref{eq:defasymadd}, we conclude that $\{n^{-1}\langle 
G_n,\IP\rangle\}$ is a Cauchy sequence and hence the limit~\eqref{4.04} exists 
and is finite. Letting $m\to\infty$ in~\eqref{eq:approxCauchy} and using 
again~\eqref{eq:defasymadd}, we conclude that the limit in~\eqref{4.04} is 
uniform in $\IP\in\PP_\varphi(M)$. Since
\begin{equation*}
\sup_{\P\in\cP_\varphi}\left|\frac1n\langle G_n,\P\rangle 
-\langle G^{(k)},\P\rangle\right|  
\leq\frac1n\left\|G_n-S_nG^{(k)}\right\|_\infty,
\end{equation*}
the uniform convergence in~\eqref{4.04} and relation~\eqref{eq:defasymadd} give 
that the convergence in~\eqref{4.04b} is also uniform in $\IP\in\PP_\varphi(M)$. 
This last uniform convergence and the continuity of $\IP\mapsto\langle 
G^{(k)},\IP\rangle$ yield the continuity of $\IP\mapsto\GG(\IP)$.
\hfill\qed
\end{proof}

The second result is a version of Birkhoff's ergodic theorem
for asymptotically additive potentials (see \cite[Proposition A.1]{FH-2010} for the case of asymptotically \textit{sub}additive potentials).

\begin{lemma}\label{lem:ergoasymadd}
Let $\P \in \cP_\varphi(M)$ and $\GG\in\aA(M)$. Then, there exists a $\varphi$-invariant function
$\overline G \in B(M)$ such that $\langle \overline G, \P \rangle = \GG(\P)$ and such that, for
$\P$-almost every $x\in M$,
\begin{equation}\label{eq:birkgn}
\lim_{n\to\infty}\frac 1 n G_n(x)  = \overline G(x).
\end{equation}
Moreover, if $\P\in \EE_\varphi(M)$, then $\overline G(x) = \GG(\P)$ for $\P$-almost every $x\in M$.
\end{lemma}
\begin{proof}{} Let $\{G^{(k)}\}$ be an approximating sequence for $\GG$. 
By Birkhoff's ergodic theorem, there exists a $\varphi$-invariant set $M_* \subset M$ with $\P(M_*) = 1$ such that for all $x\in M_*$ and all $k\geq 1$, the limit
\begin{equation}\label{eq:convGbark}
\overline G^{(k)}(x) : = \lim_{n\to\infty}\frac1nS_n G^{(k)}(x)
\end{equation}
exists.
For all $m,n,k\geq 1$ and $x\in  M_*$, we have
\begin{align*}
\left|\frac{G_n(x)}{n}-\frac{G_m(x)}{m}\right|	& \leq \left|\frac{G_n(x)}{n}-\frac{S_n G^{(k)}(x)}{n}\right| + \left|\frac{G_m(x)}{m}-\frac{S_m G^{(k)}(x)}{m}\right| \\
&\qquad \qquad +  \left|\frac{S_n G^{(k)}(x)}{n}-\frac{S_m G^{(k)}(x)}{m}\right|. 
\end{align*} 
From \eqref{eq:defasymadd} and the existence of \eqref{eq:convGbark}, we deduce that, for each $x\in M_*$, the sequence $\{n^{-1} G_n(x)\}$ is Cauchy, so that
the limit
\begin{equation}\label{eq:convGlimkGkk}
	\overline G(x) := \lim_{n\to\infty} n^{-1} G_n(x)
\end{equation}
exists. We further set $\overline G(x) = 0$ for $x\in M\setminus M_*$. As a consequence of \eqref{eq:defasymadd}, the inequality $\|S_n G^{(k)} - S_n (G^{(k)}) \circ \varphi\|_\infty \leq 2 \|G^{(k)}\|_\infty$ and the fact that $M_*$ is $\varphi$-invariant, we obtain that $\overline G$ is $\varphi$-invariant.\footnote{Notice that we also have $\overline G(x) = \lim_{k\to\infty} \overline G^{(k)}(x)$ for all $x\in M_*$.} Since moreover $\|\GG\|_* < \infty$, the family $\{n^{-1} G_n\}$ is uniformly bounded, so that $\overline G \in B(M)$ and, by the bounded convergence theorem, $\langle \overline G, \P\rangle =  \lim_{n\to\infty}\langle n^{-1}G_n, \P \rangle = \GG(\P)$. This completes the proof of the first statement.

If we further assume that $\P\in \EE_\varphi(M)$, then as a $\varphi$-invariant function, $\overline G$ must be equal to a constant $\P$-almost everywhere, and the claim then follows from \eqref{eq:birkgn}.
\hfill\qed 
\end{proof}

The third result concerns the variations of $\GG\in {\cal A}(M)$. 

\begin{lemma}\label{lem:approxGmn}
For any $\GG=\{G_n\}\in{\cal A}(M)$,\footnote{We note that~\eqref{4.010} is 
proved in~\cite[Lemma 2.1]{zhao_asymptotically_2011}.}
\begin{equation} \label{4.010}
\lim_{\delta\downarrow0}\limsup_{n\to\infty}\sup_{x\in M}\ 
\sup_{y,z\in B_n(x,\delta)}\frac1n\bigl|G_n(y)-G_n(z)\bigr|= 0.
\end{equation}
Moreover, if $m_\delta(n)$ satisfies 
$\lim_{\delta\downarrow0}\limsup_{n\to\infty}\frac1n m_\delta(n) = 0$, then
\begin{align}\label{eq:GmmmmGn}
\lim_{\delta\downarrow0}\limsup_{n\to\infty}
\frac1n\|G_{n-m_\delta(n)}-G_n\|_\infty  &= 0,\\
\lim_{\delta\downarrow0}\limsup_{n\to\infty}\sup_{x\in M}\ 
\sup_{y,z\in B_{n-m_\delta(n)}(x,\delta)} 
\frac 1 n\bigl|G_n(y)-G_n(z)\bigr| &= 0.
\label{4.010ss}
\end{align}
\end{lemma}

\begin{proof}{} 
We first prove~\eqref{4.010}. Let $\{G^{(k)}\}\subset C(M)$ be an approximating 
sequence for $\GG$. Let $\e>0$, and fix $k$ large enough so that 
$n^{-1}\bigl\|G_n-S_nG^{(k)}\bigr\|_\infty <\e$ for all $n\geq N(k)$. Let 
$\delta>0$ be such that $|G^{(k)}(x_1)-G^{(k)}(x_2)|<\e$ for 
$d(x_1,x_2)< 2\delta$. Then, for all $n\geq N(k)$, $x\in M$, and 
$y,z\in B_{n}(x,\delta)$, we have 
\begin{align*}
\bigl|G_n(y)-G_n(z)\bigr|\leq 2n\e+\bigl|S_n G^{(k)}(y)-S_n G^{(k)}(z)\bigr|  
\leq 3n\e.
\end{align*}
Since $\e>0$ is arbitrary, this establishes~\eqref{4.010}.

To prove~\eqref{eq:GmmmmGn}, fix $\e>0$ and $k,N$ large enough so that 
$\|S_{n}G^{(k)}-G_n\|_\infty\leq n\e$ for $n\geq N$. Then, for any fixed 
$\delta>0$, we have for all $n$ large enough that $n-m_\delta(n)\geq N$, and 
hence that
\begin{align*}
\|G_{n-m}-G_n\|_\infty
&\leq\|G_{n-m}-S_{n-m} G^{(k)}\|_\infty 
+\|S_{n-m}G^{(k)}-S_{n}G^{(k)}\|_\infty+\|S_{n}G^{(k)}-G_n\|_\infty\\
&\leq(n-m)\e+\|S_{n-m}G^{(k)}-S_{n}G^{(k)}\|_\infty+n\e
\leq 2n\e+m\,\|G^{(k)}\|_\infty, 
\end{align*}
where $m=m_\delta(n)$. Using the condition on $m_\delta(n)$, this gives 
\[
\lim_{\delta\downarrow0}\limsup_{n\to\infty}
\frac1n\|G_{n-m_\delta(n)}-G_n\|_\infty\leq 2\e.
\]
Since $\e>0$ is arbitrary, we obtain~\eqref{eq:GmmmmGn}.

Finally, to prove~\eqref{4.010ss}, we observe that for each $x\in M$ and 
$y,z\in B_{n-m}(x,\delta)$,
\[
\bigl|G_n(y)-G_n(z)\bigr|\leq 2\|G_{n-m}-G_n\|_\infty
+\bigl|G_{n-m}(y)-G_{n-m}(z)\bigr|.
\]
The required relation~\eqref{4.010ss} follows now from~\eqref{4.010}, 
\eqref{eq:GmmmmGn}, and the condition on $m_\delta(n)$.
\hfill\qed 
\end{proof}

\subsection{Topological pressure}\label{sec:TP}

We now introduce a notion of \textit{topological pressure\/} associated with
$\varphi$ which generalizes the usual concept of topological pressure to
asymptotically additive potential sequences. Given $\e>0$ and an integer
$n\ge1$,  a finite set~$E\subset M$ is called $(\e,n)$\textit{-separated\/} if
$y\notin B_n(x,\e)$ for any distinct $x,y\in E$, and $(\e,n)$\textit{-spanning\/}
if $\{B_n(x,\e)\}_{x\in E}$ covers~$M$. For $\GG=\{G_n\}\in\cA(M)$, $\e>0$,
and $n\ge1$ we define
\begin{align*}
S(\GG,\e,n)&=\inf\biggl\{\sum_{x\in E}\eee^{G_n(x)}: 
E\text{ is }(\e,n)\text{-spanning}\biggr\},\\
N(\GG,\e,n)&=\sup\biggl\{\sum_{x\in E}\eee^{G_n(x)}:
E\text{ is }(\e,n)\text{-separated}\biggr\}.
\end{align*}
We shall show that 
\begin{align}
\lim_{\e\downarrow0}\limsup_{n\to\infty}\frac1n\log S(\GG,\e,n)
&=\lim_{\e\downarrow0}\liminf_{n\to\infty}\frac1n\log S(\GG,\e,n),\label{2.3}\\
\lim_{\e\downarrow0}\limsup_{n\to\infty}\frac1n\log N(\GG,\e,n)
&=\lim_{\e\downarrow0}\liminf_{n\to\infty}\frac1n\log N(\GG,\e,n).\label{2.4}
\end{align}
Moreover, these two quantities coincide and their common value, which we denote
by $\fp_\varphi(\GG)$, will be called the \textit{topological pressure\/}
of~$\GG$ with respect to~$\varphi$. This result is of course well known in the
additive case $\GG=\{S_n G\}$ when $G\in C(M)$, and we shall write
\begin{equation*}
\fp_\varphi^0(G)=\fp_\varphi(\{S_n G\}).
\end{equation*}
Besides existence, we shall also establish some basic properties of
$\fp_\varphi$. They will be proven by approximation arguments, starting from
the corresponding well-known results in the additive case.
\bep\label{Prop:TopP}
\ben
\item The relations  \eqref{2.3} and \eqref{2.4}  hold and the respective quantities coincide. Moreover, the map 
\[\cA(M)\ni \GG \mapsto\fp_\varphi(\GG)\in]-\infty,+\infty]\]
is convex, and either $\fp_\varphi(\GG)=+\infty$ for all $\GG\in \cA(M)$, or
$\fp_\varphi(\GG)\in\rr$ for all $\GG\in \cA(M)$. 

\item If $\{G^{(k)}\}$ is an approximating sequence for $\GG$, then
\begin{equation} \label{4.08n}
\fp_\varphi(\GG) = \lim_{k\to\infty}\fp^0_\varphi(G^{(k)}).
\end{equation}

\item The topological entropy of $\varphi$ satisfies
\begin{equation}\label{eq:Htopsup}
h_\mathrm{Top}(\varphi)=\fp_\varphi^0(0)
=\sup_{\QQ\in\PP_\varphi(M)}h_\varphi(\QQ)
=\sup_{\QQ\in\EE_\varphi(M)}h_\varphi(\QQ).
\end{equation}

In particular, $\fp_\varphi(\GG)$ is finite for all $\GG\in \aA(M)$ if and only if $\varphi$ has finite topological entropy.
\item If the topological entropy of~$\varphi$ is finite, then for any $\GG,\GG'\in\aA(M)$ we have
\begin{equation}\label{eq:pGGGGprime}
|\fp_\varphi(\GG)-\fp_\varphi(\GG')| \leq \|\GG-\GG'\|_*.
\end{equation}
\item For any $\GG\in\aA(M)$ we have 
\begin{equation} \label{4.05}
\fp_\varphi(\GG)=\sup_{\IP\in\PP_\varphi(M)}\bigl(\GG(\IP)+h_\varphi(\IP)\bigr). 
\end{equation}
\item If the entropy map $\cP_\varphi(M)\ni\IP\mapsto h_\varphi(\IP)$ is upper 
semicontinuous, then for  any $\P\in \cP_\varphi(M)$, 
\beq
h_\varphi(\P)=\inf_{\GG\in \cA(M)}(\fp_\varphi(\GG)- \GG(\P)),
\label{Eq:hvarExpAsympt}
\eeq
and for any  $\GG \in \cA(M)$ and $\P\in \cP_\varphi(M)$ we have
\beq
h_\varphi(\P)=\inf_{G \in C(M)}(\fp_\varphi(\GG_G)- \GG_G(\P)),
\label{Eq:hvarExpAsymptre}
\eeq
where $\GG_G  = \{G_n + S_n G\}$.
\een
\eep
\begin{proof}{} 
\textit{Parts (1) and (2)\/}. As we have already mentioned, if $\GG=\{S_n G\}$ for
some $G\in C(M)$, then it is well known (see~\cite[Sections~3.1.b 
and~20.2]{KH1995} or~\cite[Section~9.1]{walters1982}) that the four quantities
in~\eqref{2.3} and~\eqref{2.4} are equal and define $\fp_\varphi^0(G)$. To
extend this result to any $\GG=\{G_n\}\in{\cal A}(M)$, we start 
with~\eqref{2.3}. Note that for any finite set $E\subset M$ and any 
$\GG'=\{G_n'\}\in{\cal A}(M)$ we have
\begin{equation}\label{sunday}
\left|\frac1n\log\sum_{x\in E}\eee^{G_n(x)}
-\frac1n\log\sum_{x\in E}\eee^{G_n'(x)}\right|
\leq\frac1n\bigl\|G_n-G_n'\bigr\|_\infty.
\end{equation}
It follows that
\begin{equation}\label{eq:SdiffGGp}
\frac1n\left|\log S(\GG,\e,n)-\log S(\GG',\e,n)\right|
\leq\frac1n\bigl\|G_n-G_n'\bigr\|_\infty.
\end{equation}
Taking  $\GG'=\{S_n G^{(k)}\}$ for some approximating sequence $\{G^{(k)}\}$ of
$\GG$, we find
\begin{equation*}
\frac1n\left|\log S(\GG,\e,n)-\log S(\{S_n G^{(k)}\},\e,n)\right|
\leq\frac1n\bigl\|G_n-S_nG^{(k)}\bigr\|_\infty.
\end{equation*}
Relation~\eqref{eq:defasymadd} and the fact that $\fp_\varphi^{0}(G^{(k)})$ is 
well defined for all $k$ give that~\eqref{2.3} holds and that the limits are 
equal to  $\lim_{k\to\infty}\fp^0_\varphi(G^{(k)})$. Relation~\eqref{sunday} 
gives that~\eqref{eq:SdiffGGp} also holds when $S(\argdot,\e,n)$ is
replaced by $N(\argdot,\e,n)$, and the above argument yields that the four 
limits in~\eqref{2.3} and~\eqref{2.4} are equal and that~\eqref{4.08n} 
holds.\footnote{See also~\cite[Lemma 1.1]{iommi_weak_2016} 
and~\cite[Lemma 2.3]{zhao_asymptotically_2011} for proofs of~\eqref{4.08n}.}

Since $G_n\mapsto\log\sum_{x\in E}\eee^{G_n(x)}$ is convex by H\"older's 
inequality, so is $\GG\mapsto \fp_\varphi(\GG)$.

By applying~\eqref{eq:nm1Gnfini} to the asymptotically additive potential 
$\GG-\GG'$, we obtain that the right-hand side of~\eqref{eq:SdiffGGp} is bounded 
uniformly in $n$. As a consequence, we have $\fp_\varphi(\GG)=\infty\iff 
\fp_\varphi(\GG')=\infty $, which yields the last statement of Part~(1).

\textit{Parts (3) and (4)\/}. Relations~\eqref{eq:Htopsup} are well known
(see~\cite[Theorem~8.6 and Theorem~9.7~(i)]{walters1982}), 
and~\eqref{eq:pGGGGprime} immediately follows from~\eqref{eq:SdiffGGp}. 

\textit{Part (5)\/}. The variational principle~\eqref{4.05} is established
in~\cite[Theorem 3.1]{FH-2010} for asymptotically \textit{sub}additive
potential sequences. We include here a proof in the (simpler) asymptotically
additive case (see also~\cite[Theorem 7.2.1]{Bar2011}).

If $h_\varphi(\IP)=+\infty$ for some $\IP\in\PP_\varphi(M)$, then 
$\fp_\varphi^0(0)=+\infty$, so that $\fp_\varphi(\GG)=+\infty$ for all 
$\GG\in\aA(M)$. Relation~\eqref{4.05} is obvious in this case. Assume now that 
$h_\varphi(\IP)<+\infty$ for all $\IP\in\PP_\varphi(M)$. 
By~\cite[Theorem~9.10]{walters1982}, for any $G\in C(M)$,
\beq
\fp^0_\varphi(G)
=\sup_{\P\in\cP_\varphi(M)}\left(h_\varphi(\P)+\langle G,\P\rangle\right).
\label{Eq:PvarExp}
\eeq
By Lemma~\ref{lem:limiteGGPP}, the sequence 
$f_k(\P):=\langle G^{(k)},\P\rangle+h_\varphi(\IP)$ converges uniformly to 
$\GG(\IP)+h_\varphi(\IP)$ on~$\cP_\varphi(M)$. It follows that
\[
\fp_\varphi(\GG)=\lim_{k\to\infty}\fp_\varphi^0(G^{(k)}) 
=\lim_{k\to\infty}\sup_{\IP\in\PP_\varphi(M)} f_k(\IP)
=\sup_{\IP\in\PP_\varphi(M)}\lim_{k\to\infty} f_k(\IP)
=\sup_{\IP\in\PP_\varphi(M)}\bigl(\GG(\IP)+h_\varphi(\IP)\bigr),
\]
where the second equality uses~\eqref{Eq:PvarExp}.

\textit{Part (6)\/}. We start by recalling (see~\cite[Theorem~9.12]{walters1982}) that 
if $\cP_\varphi(M)\ni\IP\mapsto h_\varphi(\IP)$ is upper semicontinuous, then for 
all $\P\in \cP_\varphi$,
\beq
h_\varphi(\P)=\inf_{G\in C(M)}(\fp^0_\varphi(G)-\langle G,\P\rangle).
\label{Eq:hvarExp}
\eeq
We now prove~\eqref{Eq:hvarExpAsympt} and~\eqref{Eq:hvarExpAsymptre}.
Both ``$\leq$'' inequalities are an immediate consequence of~\eqref{4.05}.

The ``$\geq$'' inequality in~\eqref{Eq:hvarExpAsympt} is immediate 
by~\eqref{Eq:hvarExp} since for $G\in C(M)$ and $\GG=\{S_n G\}$ we have 
$\fp_\varphi(\GG)-\GG(\P)=\fp^0_\varphi(G)-\langle G,\P\rangle$.
To prove the ``$\geq$'' inequality in~\eqref{Eq:hvarExpAsymptre},  
fix~$\e>0$ and use~\eqref{Eq:hvarExp} to find $W_\e\in C(M)$ such that 
\begin{equation} \label{4.014}
h_\varphi(\Q)\ge \fp^0_\varphi(W_\e)-\langle W_\e,\Q\rangle-\e. 
\end{equation}
Consider the sequence $G_k=W_\e-G^{(k)}$, where $\{G^{(k)}\}$ is an approximating 
sequence for $\GG$. In view of~\eqref{4.04b} and~\eqref{eq:pGGGGprime}, we have
\[
\GG_{G_k}(\Q)\to \langle W_\e,\Q\rangle, \quad 
\fp_\varphi(\GG_{G_k})\to \fp^0_\varphi(W_\e)
\]
as $k\to\infty$.
Combining this with~\eqref{4.014}, we see that for a sufficiently large~$k$,  
\[
\GG_{G_k}(\Q)-\fp_\varphi(\GG_{V_k})\ge
\langle W_\e,\Q\rangle-\fp_\varphi^0(W_\e)-\e\ge -h_\varphi(\Q)-2\e.
\]
Since~$\e>0$ is arbitrary, the proof is complete.\hfill\qed
\end{proof}

\subsection{Expansiveness}
\label{s2.4}

Besides the specification
properties~\hWPS and~\hS, we shall need the following assumptions 
to formulate our main results. The first concerns the regularity of the
entropy map and is needed in the proof of both Theorems~\tA and~\tB.
The second concerns the approximation of  the pressure in terms of
periodic orbits and is only needed in the proof of Theorem~\tA.
\newcommand{\hUSCE}{\hyperlink{hyp.USCE}{\textbf{(USCE)}}\xspace}
\newcommand{\hPAP}{\hyperlink{hyp.PAP}{\textbf{(PAP)}}\xspace}
\begin{description}
\item[\hypertarget{hyp.USCE}{\textbf{(USCE)}}]\label{usce} 
{\itshape $h_\mathrm{Top}(\varphi)<\infty$, and the map 
$\Q\in\PP_\varphi(M)\mapsto h_\varphi(\Q)$ is upper semicontinuous.}
\item[\hypertarget{hyp.PAP}{\textbf{(PAP)}}]\label{pap}
{\itshape For any $n\ge1$ the set $M_n$ is finite, and for any $\GG\in \cA(M)$, 
the topological pressure satisfies}
\beq
\label{2.5}
\fp_\varphi(\GG)=\lim_{n\to\infty}\frac1n\log\sum_{x\in M_n}\eee^{G_n(x)}. 
\eeq 
\end{description}
We note that, by Proposition~\ref{Prop:TopP}, Condition~\hUSCE implies that
the pressure $\fp_\varphi(\GG)$ is finite for all $\GG\in \cA(M)$ and that the 
entropy satisfies the variational principle~\eqref{Eq:hvarExpAsympt}. In 
particular, $h_\varphi(\QQ)\in[0,\infty[$ for all $\QQ\in\PP_\varphi$.

We now discuss some criteria ensuring Conditions~\hUSCE 
and~\hPAP. 
Together with the specification property, expansiveness is often considered as 
characteristic of chaotic dynamics.

\begin{definition} \label{d2.1}
The map~$\varphi$ is said to be\/ \textup{forward expansive} if there is $r>0$ such 
that if $x,y\in M$ satisfy the inequality $d(\varphi^k(x),\varphi^k(y))\le r$ for 
all $k\in\zz_+$, then $x=y$. A homeomorphism $\varphi$ is called\/
\textup{expansive} if there is $r>0$ such that if $x,y\in M$ satisfy the inequality 
$d(\varphi^k(x),\varphi^k(y))\le r$ for all $k\in\zz$, then $x=y$.
\end{definition}

The number $r$ for which this property holds is called the expansiveness (or 
expansivity) constant of~$\varphi$. 
We note that the expansiveness constant depends on the metric~$d$, whereas 
expansiveness only depends on the induced topology on~$M$. 
\bet\label{Thm:USCh}
\ben
\item If $\varphi$ is forward expansive or expansive, then Condition~\hUSCE
holds.
\item If, in addition, $\varphi$ satisfies Condition~\hWPS, then 
Condition~\hPAP holds.
\een
\eet
\begin{proof}{}
Part~(1) follows from Corollary~7.11.1 and Theorem 8.2 in~\cite{walters1982}, whose proofs 
immediately extend from case~\hH to case~\hC (see 
also~\cite[Lemma~2.4.4]{Bar2011}). 

Part~(2) is also well known in the additive case, and its proof only requires 
notational modifications in the asymptotically additive case. We include this 
proof for completeness.

Let ${\cal I}=\zz_+$ in the forward expansive case, and ${\cal I}=\zz$ in the 
expansive case. The uniform continuity of~$\varphi$ and periodicity imply that 
there is $\delta>0$ such that if $x,y\in M_n$ and $d(x,y)\le\delta$, then 
$d(\varphi^k(x),\varphi^k(y))\le r$ for all $k\in\II$, where $r$ is the 
expansivity constant. By expansiveness, such points must coincide, and 
compactness implies that~$M_n$ may contain only finitely many points.

Let $\GG\in\cA(M)$, and let $\e\in{}]0,r[$. Then, in view of periodicity, for any 
$n\ge1$ the set $M_n$ is $(\e,n)$-separated. It follows that 
\[
\sum_{x\in M_n}\eee^{G_n(x)}\le N(\GG,\e,n),
\]
whence, recalling representation~\eqref{2.4} for the pressure, we conclude 
that
\begin{equation} \label{2.6}
\limsup_{n\to\infty}\frac1n\log \sum_{x\in M_n}\eee^{G_n(x)}\le\fp_\varphi(\GG). 
\end{equation}
On the other hand, since~$\varphi$ satisfies Condition~\hWPS, the set 
$M_n$ is $(\e,n-m_\e(n))$-spanning. It follows that 
\[
\sum_{x\in M_n}\eee^{G_n(x)}\ge S(\GG,\e,n-m_\e(n)).
\]
Combining this with representation~\eqref{2.3}, and using the condition on 
$m_\e(n)$, we see that 
\begin{equation} \label{2.7}
\liminf_{n\to\infty}\frac1n\log \sum_{x\in M_n}\eee^{G_n(x)}\ge\fp_\varphi(\GG). 
\end{equation}
Inequalities~\eqref{2.6} and~\eqref{2.7} imply the required relation~\eqref{2.5}.
\hfill\qed
\end{proof}

Recall the variational principle~\eqref{4.05}. A measure~$\IP\in\PP_\varphi(M)$ 
is called an \textit{equilibrium state\/} for $\GG\in {\cal A}(M)$ if
\begin{equation} \label{2.61}
\fp_\varphi(\GG)=\GG(\IP)+h_\varphi(\IP). 
\end{equation}
We refer the reader to Section~9.5 in~\cite{walters1982} for a discussion of 
equilibrium states. The following proposition implies, in particular, that there 
always exists at least one equilibrium state if~\hUSCE and~\hPAP are assumed.
\begin{proposition} \label{p2.8}
Assume that Conditions~\hUSCE and~\hPAP hold, and let
\begin{equation*}
\overline\IP_n = \frac1n\sum_{i=0}^{n-1} \P_n \circ \varphi^{-i},
\end{equation*}
where $\IP_n$ is defined by~\eqref{add}. If~$\IP$ is a weak limit point of the 
sequence~$\{\overline\IP_n\}$, then~$\IP$ is an equilibrium state for~$\GG$. In 
particular, if~$\GG$ is additive, then $\overline\IP_n=\IP_n$, and any limit 
point of~$\{\IP_n\}$ is an equilibrium state.
\end{proposition}
\begin{proof}{} 
Let a sequence $n_k\to\infty$ be such that 
$\overline\IP_{n_k}\rightharpoonup\IP$. Since $\overline \P_n$ is invariant for 
all $n$, so is $\P$. In view of H\"older's inequality, the function 
$\GG \mapsto\log Z_n(\GG)\in\rr$ is convex. Moreover, for any 
$\GG,\GG'\in\aA(M)$, the function $\alpha \mapsto f(\alpha):=\log Z_n(\GG+\alpha\GG')$ is
differentiable, and its derivative at zero is given by 
$f'(0)=\langle G_n',\IP_n\rangle$. By convexity, we have $f(1)-f(0)\ge f'(0)$, 
which gives
\[
\log Z_{n}(\GG+\GG')-\log Z_{n}(\GG)\ge  \langle G_{n}',\IP_{n}\rangle. 
\]
Replacing $\GG=\{G_n\}$ and $\GG'=\{G_n'\}$ with $\{G_n\circ\varphi^i\}$ and 
$\{G_n'\circ\varphi^i\}$, respectively, using the relation 
$Z_n(\GG)=Z_n(\{G_n\circ\varphi^i\})$, and averaging with respect to~$i$, we 
derive 
\[
\log Z_{n}(\GG+\GG')-\log Z_{n}(\GG)\ge  \langle G_{n}',\overline\IP_{n}\rangle.
\]
We set $n=n_k$ in the above inequality, divide it by~$n_k$, and pass to the 
limit $k\to\infty$. By~\eqref{4.04} and the uniformity of the limit on invariant 
measures, the right-hand side converges to $\GG'(\IP)$, and in view 
of~\eqref{2.5}, the left-hand side converges to 
$\fp_\varphi(\GG+\GG')-\fp_\varphi(\GG)$. This leads to   the inequality 
$\fp_\varphi(\GG+\GG')-\fp_\varphi(\GG)\ge\GG'(\IP)$, which can be rewritten as
\[
\fp_\varphi(\GG+\GG')- ( \GG+\GG')(\IP)\ge\fp_\varphi(\GG)-\GG(\IP).
\]
Taking the supremum over $\GG'\in \cA(M)$ and invoking~\eqref{Eq:hvarExpAsympt}, 
we arrive at~\eqref{2.61}. The statement about the case where $\GG$ is additive 
is immediate, since then $\P_n$ is invariant.\hfill\qed
\end{proof}

\begin{remark}
None of the quantities appearing in~\eqref{2.61} depend on the specific choice
of potential $\GG$ within a given equivalence class (in the sense of
Remark~\ref{rem:equivrel}), and hence the equilibrium states depend only on the
equivalence class. The limit points of~$\{\overline\IP_n\}$, however, may depend
on the specific choice of~$\GG$ in the equivalence class. It is an interesting
question to describe potentials $\GG\in\aA(M)$ for which the invariant weak
limit points of~$\{\IP_n\}$ are equilibrium states.
\end{remark}

\section{Periodic Orbits Fluctuation Principle}
\label{s2}

\subsection{LDP for empirical measures}

Let $\GG\in{\cal A}(M)$. Recall that the sequence of probability measures
$\{\IP_n\}$ is defined by~\eqref{add}, and the sequence of  empirical
measures~$\{\mu_n^x\}$ by~\eqref{oc-me}. For a fixed $n\ge1$, we
regard~$\mu_n^x$ as a random variable on~$M$ with range in the space of
probability measures~$\PP(M)$ endowed with the weak topology and the
corresponding Borel structure.

\begin{theorem}\label{t1.9}
Suppose that Conditions~\hUSCE, \hWPS, \hS
and~\hPAP 
hold. Then: 
\ben
\item The LDP holds for~$\{\mu_n^\argdot\}$ under the laws~$\IP_n$, with the 
lower semicontinuous convex rate function $\I:\PP(M)\to[0,+\infty]$ defined by 
\begin{equation}\label{4.012first}
\I(\Q)=
\left\{
\begin{array}{cl}
-\GG(\Q)-h_\varphi(\Q)+\fp_\varphi(\GG)
&\quad\text{for }\Q\in\PP_\varphi(M),\\[4pt] 
+\infty&\quad\text{otherwise}.
\end{array}
\right.
\end{equation}
In other words,  for any Borel subset $\Gamma\subset\PP(M)$, we have 
\begin{equation}\label{2.17}
\begin{split}
-\inf_{\Q\in\bdot\Gamma}{\I}(\Q)
&\leq\liminf_{n\to\infty}n^{-1}\log\IP_n\{\mu_n^\argdot\in\Gamma\}\\[2mm]
&\le\limsup_{n\to\infty}n^{-1}\log\IP_n\{\mu_n^\argdot\in\Gamma\}
\le -\inf_{\Q\in\overline\Gamma}{\I}(\Q),
\end{split}
\end{equation}
where $\bdot\Gamma$ and~$\overline\Gamma$ stand, respectively, for the interior 
and closure of~$\Gamma$. 
\item  For any $\cV = \{V_n\}\in {\cal A}(M)$, the sequence $\frac 1 n V_n$ under 
the laws~$\P_n$ satisfies the  LDP with  the good convex rate function 
$I:\R\to[0,+\infty]$ defined  by the contraction relation\footnote{Recall that 
for~$\VV\in {\cal A}(M)$, the quantity~$\VV(\Q)$ is defined as \eqref{4.04}.}
\begin{equation*}
I(s)=\inf\{\I(\Q): \Q\in\PP_\varphi(M),\VV(\Q)=s\}.
\end{equation*}
In other words,  for any Borel subset $\Gamma\subset\rr$, we have 
\begin{equation}\label{2.177}
\begin{split}
-\inf_{s\in\bdot\Gamma}I(s)
&\le\liminf_{n\to\infty}n^{-1}\log\IP_n\left\{n^{-1}V_n \in\Gamma\right\}\\[2mm]
&\le\limsup_{n\to\infty}n^{-1}\log\IP_n\{n^{-1}V_n \in\Gamma\}
\le -\inf_{s\in\overline\Gamma}I(s).
\end{split}
\end{equation}
Moreover, $I$ is the Legendre transform of the function 
$\alpha\mapsto\fp_\varphi(\GG+\alpha\VV)-\fp_\varphi(\GG)$.
\een
\end{theorem}

\begin{remark} An immediate consequence of~\eqref{4.012first} is that the 
function $\I$ is affine on $\cP_\varphi(M)$. We observe also that $\I(\Q) = 0$ iff $\Q$ is an equilibrium state for $\GG$.
\end{remark}

\begin{remark}\label{rem:additivecaseknown}
In the additive case $\GG=\{S_nG\}$ and $\VV=\{S_n V\}$ for some $G,V\in C(M)$,
Theorem~\ref{t1.9} is a direct consequence of known results: Part~(1) follows
from~\cite[Theorem~5.2]{comman-2009} while Part~(2) is obtained from Part~(1) by
an application of the contraction principle~\cite[Theorem 4.2.1]{DZ2000}. When
$\GG$ is only asymptotically additive, Part~(1) is new and an approximation
argument is required to deduce Part~(2) from the contraction principle when
$\VV$ is only asymptotically additive (see Section~\ref{ss:LDPlev1}). This
applies, in particular, to the entropy production sequence $\{\sigma_n\}$
defined by~\eqref{Eq:SigmaDef}.
\end{remark}

Obviously, Part~(1) of Theorem~\ref{t1.9} yields the LDP part of 
Theorem~\tA in the asymptotically additive setting.

The proof of Theorem~\ref{t1.9} is postponed to Section~\ref{s:abstractLevel3}, 
more precisely to Theorem~\ref{appt1.17}. There some more general measures $\P_n$ 
are considered, in order to give a unified treatment to the measures $\P_n$ 
in~\eqref{add} and the weak Gibbs measures considered in Section~\ref{s4}.

\subsection{Fluctuation Theorem and Fluctuation Relations }
\label{s2.3}

In this subsection we prove the FT and FR parts of Theorem~\tA in the general 
setting of the previous subsection. To this end, in addition to 
Conditions~\hUSCE, \hWPS, \hS and~\hPAP which 
are needed for the LDP, we impose one of the following assumptions to ensure  the 
validity of FR.
\newcommand{\hCom}{\hyperlink{hyp.Com}{\textbf{($\boldsymbol{\mathcal{C}}$)}}\xspace}
\newcommand{\hRev}{\hyperlink{hyp.Rev}{\textbf{($\boldsymbol{\mathcal{R}}$)}}\xspace}
\begin{description}
\item[\hypertarget{hyp.Com}{\textbf{($\boldsymbol{\mathcal{C}}$-Commutation)}}]
\label{commutation}
{\itshape There is a homeomorphism $\theta:M\to M$ such that $\theta\circ\theta=\Id_M$ and 
$\varphi=\theta\circ\varphi\circ\theta$.}
\item[\hypertarget{hyp.Rev}{\textbf{($\boldsymbol{\mathcal{R}}$-Reversal)}}]
\label{reversal} 
{\itshape The map $\varphi$ is a homeomorphism and there is a homeomorphism 
$\theta:M\to M$ such that $\theta\circ\theta=\Id_M$ and 
$\varphi^{-1}=\theta\circ\varphi\circ\theta$.}
\end{description}
Let us remark that in both cases, if $\Q\in\PP_\varphi(M)$, then 
$\widehat\Q:=\Q\circ\theta\in\PP_\varphi(M)$. Indeed, in case~\hRev, for any $V\in C(M)$ we have
\[
\langle V\circ\varphi,\widehat\Q\rangle
=\langle V\circ\varphi\circ\theta,\Q\rangle
=\langle V\circ\theta\circ\varphi^{-1},\Q\rangle
=\langle V\circ\theta,\Q\rangle
=\langle V,\widehat\Q\rangle.
\]
A similar argument applies in case~\hCom.

Condition~\hRev is the standard dynamical system reversal condition that 
appears in virtually all works on the FT and FR. 
To the best of our knowledge, it was not previously remarked that~\hCom
also suffices to derive the FR. Since~\hCom does \textit{not} 
require that~$\varphi$ be a homeomorphism, it allows one to expand the class of 
examples for which the FR can be established. 
\begin{example}  \label{e2.13}
Let $X$ be a compact metric space and $\phi: X\rightarrow X$ a continuous map. 
Set $M=X\times X$ and $\theta(x,y)=(y,x)$. Then~\hCom holds 
for the map~$\varphi:(x,y)\mapsto(\phi(x),\phi(y))$. If~$\phi$ is a 
homeomorphism, then~\hRev holds for the  
homeomorphism~$\varphi:(x,y)\mapsto(\phi(x),\phi^{-1}(y))$.
\end{example} 
\begin{example}\label{interval}
An interesting concrete setting of Example~\ref{e2.13} is the case $X=[0,1]$.
The classical examples of interval maps $\phi:[0,1]\rightarrow [0,1]$ such as
the tent map, Farey map,  or the Pomeau-Manneville map are not bijections 
and~\hRev cannot hold whereas~\hCom applies. We refer the reader
to~\cite{CJPS_phys} for a  discussion of the  FT and FR  for this class of 
examples.
\end{example}

Observe that Condition~\hRev/\hCom implies
\begin{equation}
S_n(V\circ\theta)(x)=\sum_{0\le k<n}V\circ\theta\circ\varphi^k(x)
=\sum_{0\le k<n}V\circ\varphi^{\mp k}\circ\theta(x)
=(S_n V)(\theta_n^\mp(x))\label{eq:SnVtheta}	
\end{equation}
where 
\[
\theta_n^-=\theta\circ\varphi^{n-1}\text{ in the case }~\hRev,  
\qquad  
\theta_n^+=\theta\text{ in the case }~\hCom.
\]
Note  that the map $\theta_n^\pm$ is  an involution: in the case~\hCom this is 
immediate, and in the case~\hRev we have 
\[
(\theta_n^-)^{-1}=\varphi^{1-n}\circ\theta^{-1}=\varphi^{1-n}\circ\theta
=\theta\circ\varphi^{n-1}=\theta_n^-.
\]
For $\GG=\{G_n\}\in\aA(M)$, we define $\GG\circ\theta=\{G_n\circ\theta_n^\pm\}$.
\begin{lemma}\label{l:Qhatasym}
If $\GG$ is asymptotically additive with approximating sequence $\{G^{(k)}\}$, 
then $\GG\circ\theta$ is asymptotically additive with approximating sequence 
$\{G^{(k)}\circ\theta\}$.
\end{lemma}
\begin{proof}{} 
By~\eqref{eq:SnVtheta}, we have
\begin{align*}
\bigl\|G_n\circ\theta^\pm_n-S_n(G^{(k)}\circ\theta)\bigr\|_\infty
=\bigl\|G_n\circ\theta^\pm_n-S_n(G^{(k)})\circ\theta^\pm_n \bigr\|_\infty
=\bigl\|G_n-S_nG^{(k)}\bigr\|_\infty, 
\end{align*}
and the result follows.\hfill\qed
\end{proof}

It follows from Lemmas~\ref{lem:limiteGGPP} and~\ref{l:Qhatasym} that
\begin{equation}\label{4.019}
(\GG\circ\theta)(\Q)=\lim_{k\to \infty}\langle G^{(k)}\circ\theta,\Q\rangle
=\lim_{k\to\infty}\langle G^{(k)},\widehat\Q\rangle
= \GG(\widehat\Q),\quad \Q\in\P_\varphi(M).
\end{equation}
In addition, we have 
\begin{lemma}
The following holds: 
\begin{align}
h_\varphi(\widehat\Q)
&=h_\varphi(\Q)\quad\mbox{for $\Q\in\P_\varphi(M)$},\label{Eq:htheta}\\[2mm]
\fp_\varphi(\GG\circ\theta)
&=\fp_\varphi(\GG)\quad\text{for }\GG\in\cA(M).\label{eq:Psymtheta}	
\end{align}
\end{lemma}
\begin{proof}{}  
Let $\psi=\theta\circ\varphi\circ\theta$. Since the Kolmogorov--Sinai entropy is
conjugacy invariant (see~\cite[Theorem 4.11]{walters1982}), 
$h_\psi(\Q)=h_\varphi(\Q\circ \theta)=h_\varphi(\widehat\Q)$. 
In case~\hCom we have $\psi=\varphi$, and~\eqref{Eq:htheta} follows. In 
case~\hRev we note that by~\cite[Theorem~4.13]{walters1982}, 
$h_\psi(\Q)=h_{\varphi^{-1}}(\Q)=h_\varphi(\Q)$. Thus~\eqref{Eq:htheta} holds in 
both cases.

In order to prove~\eqref{eq:Psymtheta}, we observe that, by~\eqref{4.05}, 
\eqref{4.019} and~\eqref{Eq:htheta}, 
\begin{align*}
\fp_\varphi(\GG\circ\theta)
&=\sup_{\Q\in\PP_\varphi(M)}\bigl((\GG\circ \theta)(\Q)+h_\varphi(\Q)\bigr)
=\sup_{\Q\in\PP_\varphi(M)}\bigl(\GG(\widehat\Q)+h_\varphi(\widehat\Q)\bigr)\\
&=\sup_{\Q\in\PP_\varphi(M)}\bigl(\GG(\Q)+h_\varphi(\Q)\bigr)
=\fp_\varphi(\GG),
\end{align*}
which completes the proof.\hfill\qed
\end{proof}

The \textit{entropy production in time~$n$\/} is defined by 
\begin{equation}\label{4.31}
\sigma_n=\sigma_n[\GG]=G_n-G_n \circ\theta_n^\pm. 
\end{equation}
Observing  that~$M_n$ is strictly invariant under $\theta_n^\pm$, the following 
result is immediate.
\begin{lemma}
Let $\GG\in\cA(M)$. Then $Z_n(\GG)=Z_n(\GG\circ\theta)=:Z_n$. Moreover, 
$\IP_n\circ\theta_n^\pm$ is the measure given by~\eqref{add} for the potential 
sequence $\GG\circ\theta$, \ie
\begin{equation*}
\IP_n\circ\theta_n^\pm
=Z_n^{-1}\sum_{x\in M_n}\eee^{G_n\circ\theta_n^\pm(x)}\delta_x.	
\end{equation*}
Finally, the measures~$\IP_n$ and~$ \IP_n \circ\theta_n^\pm$ are 
equivalent, and we have 
\begin{equation*}
\log\frac{\dd\IP_n}{\dd\IP_n\circ\theta_n^\pm}=\sigma_n.
\end{equation*}
\end{lemma}
Condition~\hPAP implies that the limit defining the entropic pressure
\[
e(\alpha)=\lim_{n\to\infty}\frac{1}{n}
\log\int_{M_n}\eee^{-\alpha\sigma_n}\d\P_n
\]
exists for all $\alpha\in\rr$ and is given by
\[
e(\alpha)=\fp_\varphi\left((1-\alpha)\GG+\alpha\,\GG\circ\theta\right) - \fp_\varphi(\GG) .
\]

For $\Q\in\cP_\varphi(M)$, let 
\begin{equation}
\ep(\Q)=\GG(\Q)-(\GG\circ\theta)(\Q).
\end{equation}

Note that 
\begin{equation}\label{eq:symmetryep}
\ep(\widehat \Q) = -\ep(\Q),
\end{equation}
and that if $G_n=S_n V$ and $\Q\in\cP_\varphi(M)$, then 
$\ep(\Q)=\langle V-V\circ\theta,\Q\rangle=\langle V,\Q-\widehat\Q\rangle$. 

\begin{theorem}\label{t1.14}
In addition to the hypotheses of Theorem~\ref{t1.9}, suppose that 
Condition~\hCom/\hRev holds. Then the rate function~$\I$ of the LDP 
for the empirical measures~\eqref{oc-me} under the laws~$\IP_n$ satisfies the
relation\footnote{Note that~\eqref{2.36} implies~\eqref{1.8} for $\Q\in
\cP_\varphi(M)$. On the other hand, when $\Q$ is not invariant, both sides 
of~\eqref{1.8} are $+\infty$, whereas in the asymptotically additive setup, 
the quantity $\ep(\Q)$ is defined only for $\Q\in\cP_\varphi(M)$.}
\begin{equation}\label{2.36}
\I(\widehat\Q)=\I(\Q)+\ep(\Q)
\end{equation}
for any $\QQ\in\PP_\varphi(M)$. 
Furthermore, the sequence  $\frac1n\sigma_n$ under the laws $\P_n$  satisfies the 
LDP~\eqref{2.177} with the good convex rate 
function given by 
\beq\label{2.38}
I(s)=\inf\{\I(\Q):\Q\in\PP_\varphi(M), \ep(\Q)=s\}.
\eeq
The rate function $I$ satisfies the relation 
\beq 
I(-s)=I(s)+s
\label{FRagain}
\eeq
for all $s\in\rr$ and  is the Legendre transform of the function $\alpha \mapsto e(-\alpha)$.
\end{theorem}

\begin{proof}{} Recall that the rate function~$\I$ is given 
by~\eqref{4.012first}. By~\eqref{4.019} and~\eqref{Eq:htheta}, for any 
$\Q\in\PP_\varphi(M)$,
\[
\I(\widehat\Q)=-\GG(\widehat\Q)-h_\varphi(\widehat\Q)+\fp_\varphi(\GG)
=-(\GG\circ\theta)(\Q)-h_\varphi(\Q)+\fp_\varphi(\GG),
\]
and~\eqref{2.36} follows. 

To obtain the LDP for $n^{-1} \sigma_n$, we first observe that, by 
Lemma~\ref{l:Qhatasym}, the sequence  $\{\sigma_n\}$ is asymptotically additive 
with approximating sequence $\sigma^{(k)} = G^{(k)} - G^{(k)} \circ \theta$. The 
LDP with rate function~\eqref{2.38} then follows from Part~(2) of 
Theorem~\ref{t1.9}.

Finally, the FR~\eqref{FRagain} follows from Proposition~\ref{p4.3}, but 
it also can be directly deduced as follows. 
Combining~\eqref{eq:symmetryep} with~\eqref{2.36} and~\eqref{2.38},
we see that 
\begin{align*}
I(-s)&=\inf\{\I(\Q):\Q\in\PP_\varphi(M), \ep(\Q)=-s\} \\[2mm]
&=\inf\{\I(\widehat\Q)-\ep(\Q):\Q\in\PP_\varphi(M), \ep(\Q)=-s\}\\[2mm]
&=\inf\{\I(\Q'):\Q'\in\PP_\varphi(M),\ep(\Q')=s\}+s
=I(s)+s,
\end{align*}
where we used the fact that $\Q\in\PP_\varphi(M)$ if and only if $\widehat \Q\in\PP_\varphi(M)$.
This completes the proof of Theorem~\ref{t1.14}.\hfill\qed 
\end{proof}

\section{Weak Gibbs measures}
\label{s4}

\begin{definition}\label{d1.16}
We say that~$\IP\in\PP(M)$ is a \textup{weak Gibbs measure} for $\GG\in \cA(M)$ if 
for any $n\ge1$ and any $\e> 0$ there is $K_n(\e)\ge1$ such that
\begin{gather}
K_n(\e)^{-1}\,\eee^{G_n(x)-n\fp_\varphi(\GG)}\le \IP\bigl(B_n(x,\e)\bigr)
\le K_n(\e)\,\eee^{G_n(x)-n\fp_\varphi(\GG)},
\label{2.47} \\[6pt]
\lim_{\e\downarrow0}\limsup_{n\to\infty}\frac1n\log K_n(\e)=0,
\label{2.48}
\end{gather}
where~\eqref{2.47} holds for every $x\in M$.
\end{definition}
\begin{remark}
It is easy to see that if a probability measure $\P\in \cP(M)$ is weak Gibbs
for two potential sequences  $\GG$ and $\GG'\in \cA(M)$, then 
$\{G_n-n\fp_\varphi(\GG)\}_{n\geq 1}$ and $\{G_n'-n\fp_\varphi(\GG')\}_{n\geq 1}$ are
equivalent in the sense of Remark~\ref{rem:equivrel}. Conversely, if 
$\GG\sim\GG'$, then~$\P$ is weak Gibbs for~$\GG$ iff it is weak Gibbs for $\GG'$.
\end{remark}

We emphasize that the definition of weak Gibbs measure does not require
$\IP\in\cP_\varphi(M)$. Notice also that it follows from~\eqref{2.47} that the
support of~$\IP$ coincides with~$M$. The following lemma shows that if the
latter property is satisfied, then it suffices to require the validity
of~\eqref{2.47} almost everywhere. This observation is technically useful when
transfer operators are used to construct weak Gibbs measures; 
see~\cite[Section 2]{kessebohmer-2001}, \cite[Appendix B]{climenhaga2010}, 
and~\cite{CJPS_phys}.

\begin{lemma}
Let $\GG=\{G_n\}\in\aA(M)$, and let $\IP\in\PP(M)$. Assume that for all $\e>0$, there is a dense set $A_\e$ such that \eqref{2.47} holds for
all $x\in A_\e$ and $n\geq 1$, with~$\{K_n(\e)\}$ satisfying~\eqref{2.48}. Then
$\IP$ is weak Gibbs for $\GG$.
\end{lemma}
\begin{proof}{} Define
\begin{equation} 
\gamma(n,\e)=\sup_{x\in M}\sup_{y,z\in B_n(x,\e)} 
\frac1n\bigl|G_n(y)-G_n(z)\bigr|.
\label{eq:gammandelta}
\end{equation}
By~\eqref{4.010},  
\begin{equation}\label{4.004}
\lim_{\e \downarrow0}\limsup_{n\to\infty}\gamma(n,\e)=0.	
\end{equation}
Let  $n\geq 1$, $\e>0$, and fix $x\in M$. Since 
$\bar A_\e=M$, we can find $x'\in A_\e\cap B_n(x,\e/2)$, so that
\begin{align*}
\P(B_n(x, \e))&\geq\P(B_n(x', \e/2)) 
\geq K_n(\e/2)\,\eee^{G_n(x')-n\fp_\varphi(\GG)} 
\geq K_n'(\e)^{-1}\eee^{G_n(x)-n\fp_\varphi(\GG)},\\[2mm]
\P(B_n(x, \e))&\leq\P(B_n(x',2\e)) 
\leq K_n(2\e)\eee^{G_n(x')-n\fp_\varphi(\GG)} 
\leq K_n'(\e)\eee^{G_n(x)-n\fp_\varphi(\GG)},
\end{align*}
where 
\begin{equation*}
K_n'(\e)=\max(K_n(2\e), K_n(\e/2))\,\eee^{n\gamma(n,\e)}.
\end{equation*}
Relation~\eqref{4.004} gives that 
\[
\lim_{\e\downarrow 0}\limsup_{n\rightarrow\infty}\frac{1}{n}\log K_n'(\e)=0,
\]
and the statement follows.\hfill\qed
\end{proof}

We next show that invariant weak Gibbs measures are equilibrium states
(the converse is not true in general).

\begin{lemma}\label{lem:WGestequil}
Let $\GG=\{G_n\}\in\aA(M)$ and assume $\P$ is a $\varphi$-invariant
weak Gibbs measure for $\GG$. Then $\P$ is an equilibrium state for $\GG$.
\end{lemma}

\begin{proof}{}
By \eqref{2.47}, we have for all $x\in M$ that
\begin{equation} \label{2.102}
\lim_{\e\downarrow0}\limsup_{n\to\infty}\frac1n
\bigl(-\log\IP(B_n(x,\e))+G_n(x)\bigr)=\fp_\varphi(\GG).
\end{equation}
On the other hand, the
Brin--Katok local entropy formula and Lemma~\ref{lem:ergoasymadd} imply that for $\IP$-a.e.~$x\in M$,
\begin{equation} \label{2.103}
\lim_{\e\downarrow0}\limsup_{n\to\infty}\frac1n
\bigl(-\log \IP(B_n(x,\e))+G_n(x)\bigr)
=h_\varphi(\IP,x)+\overline{G}(x),
\end{equation}
where $h_\varphi(\IP,\argdot)$ and $\overline{G}$ are $\varphi$-invariant 
functions such that $\langle h_\varphi(\IP,\argdot),\IP\rangle=h_\varphi(\IP)$ and $\langle\overline{G},\IP\rangle=\GG(\IP)$.
Equating the right-hand sides of~\eqref{2.102} and~\eqref{2.103} and integrating 
with respect to~$\IP$, we arrive at~\eqref{2.61}.
\hfill\qed 
\end{proof}

The following result is again a special case of Theorem~\ref{appt1.17}.
\begin{theorem}\label{t1.17}
Assume that Conditions~\hUSCE and~\hS hold and that $\P$ is a weak  
Gibbs measure for $\GG\in\cA(M)$. Then the conclusions of Theorem~\ref{t1.9} hold 
with~$\IP_n$ replaced by~$\IP$.
\end{theorem}

\begin{remark}\label{rem:commentt44} The same comment as in Remark~\ref{rem:additivecaseknown} applies to
Theorem~\ref{t1.17}: in the additive case, the LDP for the empirical measures
is known (\cite[Theorem~5.2]{comman-2009}, \cite[Theorem~1]{PS-2018}) and
the LDP for $n^{-1} \sigma_n$ follows by the contraction principle.
\end{remark}

Recall that $\sigma_n$ is defined by~\eqref{4.31}. By Lemma~\ref{lem:WGPress} 
below, the limit
\[
e(\alpha):=\lim_{n\to\infty}\frac{1}{n}
\log\int_{M}\eee^{-\alpha \sigma_n}\d\P
\]
exists for all $\alpha\in\rr$ and is given by
\[
e(\alpha)=\fp_\varphi\left((1-\alpha)\GG+\alpha\,\GG\circ\theta\right)  - \fp_\varphi(\GG) .
\]
The proof of the following result is exactly the same as that of 
Theorem~\ref{t1.14}.
\begin{theorem}
If in  addition to the hypotheses of Theorem~\ref{t1.17}, 
Condition~\hCom/\hRev holds, then all the conclusions of 
Theorem~\ref{t1.14} hold with $\P_n$ replaced by $\P$.
\end{theorem}

\begin{remark}\label{rem:atequilibrium2}
Our assumptions do not exclude the situation
where the sequence $\{n^{-1} \sigma_n\}$ is exponentially equivalent to 0, \ie
the situation where $I(0) = 0$ and $I(s) = +\infty$ for all $s\neq 0$. In this case
the symmetry \eqref{FRagain} is trivial, and $\ep(\P) = 0$
(note that \eqref{FRagain} always implies that $\ep(\P) \geq 0$).
It is a non-trivial question to determine whether a given system
is truly out of equilibrium, \ie if $\ep(\P) > 0$.
In the additive case, when $G_n = S_n G$, and if $\P \in \cP_\varphi(M)$,
we have $\ep(\P) = \int \sigma \d \P$, where $\sigma$ is given by \eqref{eq:defsigmaGG}.
In the context of Example~\ref{exa:anosovhisto}, when $\P$ is the SRB measure,
the inequality $\ep(\P) > 0$ is called {\em dissipativity} and is often {\em assumed}
(see~\cite{BGM_EFNC} for a model where dissipativity can be proved 
and~\cite[Proposition~11.5]{JPRB-2011} for generic results in this direction).
See also \cite{ruelle1996positivity,BFGGUChaos} for discussions and
\cite[Theorem 5.2]{MV-2003}, \cite[Corollary 10.16 (7)]{JPRB-2011}.
\end{remark}

\begin{remark} Weak Gibbs measures have been extensively studied in the recent 
literature on multifractal formalism; see~\cite{CJPS_phys} for references and 
additional information.
\end{remark}

\section{Large deviation principles}
\label{s:abstractLevel3}

\subsection{Main result and applications}

In this subsection we establish the LDP for the empirical measures~$\{\mu_n^x\}$
defined in~\eqref{oc-me} and for asymptotically additive potential sequences 
under some assumptions that cover both the sequence of measures~\eqref{add}
concentrated on periodic orbits and weak Gibbs measures.

We begin with the identification of the rate function for the LDP. Given a 
sequence $\GG=\{G_n\}\in{\cal A}(M)$ and~$V\in C(M)$, we set  
$\GG_V=\{G_n+S_n V\}$ and define a map~$\I:\PP(M)\to[0,\infty]$ by
\begin{equation}\label{4.011}
\I(\Q)=\sup_{V\in C(M)}\bigl(\langle V,\Q\rangle
-\fp_\varphi(\GG_V)+\fp_\varphi(\GG)\bigr).
\end{equation}
Note that it follows from~\eqref{eq:pGGGGprime} that for any fixed $\GG$ we 
have 
\[
|\fp_\varphi(\GG_V)-\fp_\varphi(\GG_{V'})| \leq \|V-V'\|_\infty,
\]
and in particular that the map $V\mapsto\fp_\varphi(\GG_V)$ is continuous.

\begin{lemma}
The map~$\cP(M)\ni\Q\mapsto\I(\Q)\in[0,+\infty]$ is lower semicontinuous and 
convex. Moreover, if Condition~\hUSCE holds, then
\begin{equation} \label{4.012}
\I(\Q)=
\left\{
\begin{array}{cl}
-\GG(\Q)-h_\varphi(\Q)+\fp_\varphi(\GG)
&\quad\text{for }\Q\in\PP_\varphi(M),\\[4pt] 
+\infty&\quad\text{otherwise}.
\end{array}
\right.
\end{equation}
\end{lemma}
An immediate consequence of~\eqref{4.012} is that the map 
$\cP_\varphi(M)\ni\Q\mapsto\I(\Q)$ is affine if~$\varphi$ satisfies~\hUSCE.

\begin{proof}{}
By definition~\eqref{4.011}, the function~$\I$ is the pointwise supremum of a 
family of continuous affine maps. Therefore it is convex and lower semicontinuous.

If $\Q\notin\PP_\varphi(M)$, then there is $V\in C(M)$ such that
$\delta:=\langle V,\Q\rangle-\langle V\circ\varphi,\Q\rangle>0$. Thus, letting
$V_m=m(V-V\circ\varphi)$ and observing that 
$\|S_n V_m\|_\infty\le 2m\|V\|_\infty$, we deduce from~\eqref{eq:pGGGGprime} that
$\fp_\varphi(\GG_{V_m})=\fp_\varphi(\GG)$. Since 
$\langle V_m,\Q\rangle=m\delta$, we conclude that the supremum in~\eqref{4.011} 
is equal to~$+\infty$.

To prove~\eqref{4.012} in the case $\Q\in\PP_\varphi(M)$, we 
rewrite~\eqref{4.011} as 
\[
\I(\Q)=\sup_{V\in C(M)}\bigl(\GG_V(\Q)-\fp_\varphi(\GG_V)\bigr) 
-\GG(\QQ)+\fp_\varphi(\GG).
\]
The required result now follows by the variational 
principle~\eqref{Eq:hvarExpAsymptre}.\hfill\qed
\end{proof}

Given a sequence  $\{\P_n\}_{n\geq 1}\subset \cP(M)$ and a function~$V\in C(M)$, 
we define 
\[
A_n(V)=\int_M\eee^{n\langle V,\mu_n^x\rangle}\IP_n(\dd x)
=\int_M\eee^{S_n V(x)}\IP_n(\dd x). 
\]
\newcommand{\hCone}{\hyperlink{hyp.C1}{\textbf{(C1)}}\xspace}%
\newcommand{\hCtwo}{\hyperlink{hyp.C2}{\textbf{(C2)}}\xspace}%
\begin{theorem}\label{appt1.17}
Suppose that Conditions~\hS and~\hUSCE hold. Let 
$\{\P_n\}\subset\cP(M)$ and  $\GG\in\aA(M)$ be such that:
\begin{itemize}
\item[\hypertarget{hyp.C1}{\textbf{(C1)}}] For all $V\in C(M)$, we have
\begin{equation} \label{2.03first}
\lim_{n\to\infty}\frac1n\log A_n(V)=\fp_\varphi(\GG_V) - \fp_\varphi(\GG).
\end{equation}
\item[\hypertarget{hyp.C2}{\textbf{(C2)}}] For any $0<\delta\ll1$ there is an integer $n_0(\delta)\ge1$ and 
sequences $K_n(\delta)\ge1$,  $m_\delta(n)\in \N$, such that
\begin{gather}
K_n(\delta)^{-1}\,\eee^{G_n(x)-n\fp_\varphi(\GG)}
\le \IP_n\bigl(B_{n-m_\delta(n)}(x,\delta)\bigr)
\quad\text{for }x\in M, n\geq n_0(\delta),\label{eq:applowerassum}\\
\lim_{\delta\downarrow0}\limsup_{n\to\infty}
\frac1n\bigl(\log K_n(\delta)+ m_\delta(n)\bigr)=0.
\label{eq:conditionKnMdelta}
\end{gather}
\end{itemize}
Then: 
\ben
\item The LDP~\eqref{2.17} holds with the rate function $\I$ given 
by~\eqref{4.012}.
\item For any $\cV=\{V_n\}\in{\cal A}(M)$, the sequence $\{\frac1nV_n\}$ satisfies 
the LDP~\eqref{2.177} with the good convex rate function $I:\R\to[0,+\infty]$ 
defined by the contraction relation  
\begin{equation*}
I(s)=\inf\{\I(\Q):\Q\in\PP_\varphi(M),\VV(\Q)=s\}.
\end{equation*}
Moreover, $I$ is the Legendre transform of the function
$\alpha\mapsto\fp_\varphi(\GG+\alpha \VV)-\fp_\varphi(\GG)$.
\een
\end{theorem}
The two parts of Theorem~\ref{appt1.17} are proved in Sections~\ref{subs:LDPOM}
and~\ref{ss:LDPlev1} below. Here we prove that Theorem~\ref{appt1.17} applies to
the sequences of measures~\eqref{add} and to weak Gibbs measures.

\begin{lemma}
Assume that $\varphi$ satisfies Conditions~\hWPS and~\hPAP. Then the 
measures defined by~\eqref{add} satisfy Conditions~\hCone and~\hCtwo.
\end{lemma}
\begin{proof}{} Since 
\begin{equation*}
\log A_n=\log\sum_{x\in M_n}\eee^{G_n(x)+S_n V(x)}-\log Z_n 
=\log\sum_{x\in M_n}\eee^{G_n(x)+S_n V(x)}-\log\sum_{x\in M_n}\eee^{G_n(x)},
\end{equation*}
Condition~\hPAP yields~\hCone.

To prove~\hCtwo,  let $m_\delta(n)$ and $n_0(\delta)$ be as in the definition of 
Condition~\hWPS. Set
\[
\lambda(n,\delta) = \frac 1 n\|G_{n-m_\delta(n)} - G_n\|_\infty,
\]
and let~$\gamma(n,\delta)$ be  defined by~\eqref{eq:gammandelta}.
It follows from Lemma~\ref{lem:approxGmn} that
\begin{equation} \label{5.8}
\lim_{\delta \downarrow 0 }\limsup_{n\to\infty} 
\bigl(\gamma(n,\delta) +\lambda(n,\delta)\bigr)= 0.	
\end{equation}
By \hWPS, for any $n\geq n_0(\delta)$ the intersection 
$M_n\cap B_{n-m_\delta(n)}(x,\delta)$ contains at least one point~$x_n$, so that, by 
writing $n'(\delta)=n-m_\delta(n)$, we have
\begin{align*}
\log\IP_n\bigl(B_{n'(\delta)}(x,\delta)\bigr)
&\geq G_{n'(\delta)}(x_n)-\log Z_{n'(\delta)}\\
&\geq G_{n'(\delta)} (x)  - n'(\delta)\gamma(n'(\delta),\delta)- \log Z_{n'(\delta)}\\
&\geq G_n(x)-n\lambda(n,\delta)-n'(\delta)\gamma(n'(\delta),\delta)-\log Z_{n'(\delta)}\\
&\geq G_n(x)-n\fp_\varphi(\GG)-\log K_n(\delta),
\end{align*}
where
\[
K_n(\delta)=\exp\bigl\{n\lambda(n,\delta)+n'(\delta)\gamma(n'(\delta),\delta)
+|\log Z_{n'(\delta)}-n\fp_\varphi(\GG)|\bigr\}. 
\]
Condition~\hCtwo now follows from~\eqref{5.8}, the definition of $\fp_\varphi(\GG)$, and the condition satisfied by $m_\delta(n)$ in \hWPS. 
\hfill\qed 	
\end{proof}

\begin{lemma}\label{lem:WGPress}
Suppose that Condition~\hUSCE holds and that $\P\in\PP(M)$ is a  weak Gibbs measure 
for $\GG\in\aA(M)$. Then for all $\GG'\in\cA(M)$ we have
\begin{equation}\label{2.03zero}
\lim_{n\to\infty}\frac1n\log \left\langle\eee^{G_n'},\IP\right\rangle
=\fp_\varphi(\GG+\GG')-\fp_\varphi(\GG).
\end{equation}
\end{lemma}
\begin{proof}{} 
The proof follows that of Proposition~3.2 in~\cite{kifer-1990} (see 
also~\cite[Proposition~10.3]{JPRB-2011}); since an additional limiting argument is needed, 
we include it for the sake of completeness.

For any $\e>0$, $n\ge1$, and any $(\e,n)$-spanning set $E_{\e,n}$, using~\eqref{2.47} we 
derive
\begin{align*}
\left\langle\eee^{G_n'},\IP \right\rangle
&\le\sum_{x\in E_{\e,n}}\bigl\langle\boldsymbol{1}_{B_n(x,\e)}\eee^{G_n'},\IP\bigr\rangle 
\leq\eee^{n \gamma(n,\e)}\sum_{x\in E_{\e,n}}
\eee^{G_n'(x)}\IP\bigl(B_n(x,\e)\bigr)\\
&\leq K_n(\e)\eee^{n \gamma(n,\e)-n\fp_\varphi(\GG)}\sum_{x\in E_{\e,n}}
\eee^{G_n(x)+G_n'(x)},
\end{align*}
where $\gamma(n,\e)$ is defined by~\eqref{eq:gammandelta}. It follows that 
\begin{equation*}
\limsup_{n\to\infty}\frac1n\log\left\langle\eee^{G_n'},\IP\right\rangle 
\leq\limsup_{n\to\infty}\bigl(n^{-1}\log K_n(\e)+\gamma(n,\e)
+n^{-1}\log S(\GG+\GG',\e,n)\bigr)-\fp_\varphi(\GG).
\end{equation*}
Sending $\e\to 0$ and using expression~\eqref{2.3} for $\fp_\varphi(\GG+\GG')$, we obtain 
the ``$\leq$'' inequality in~\eqref{2.03zero}. 

To prove the reverse inequality, we proceed similarly, observing that for any 
$(\e,n)$-separated set $E_{\e,n}$ we have
\begin{align*}
\left \langle \eee^{G_n'}, \IP \right\rangle
&\geq \sum_{x\in E_{\e,n}}\bigl\langle\boldsymbol{1}_{B_n(x,\e/2)}\eee^{G_n'},\IP\bigr\rangle  \geq K_n^{-1}(\e/2)\eee^{-n\gamma(n,\e/2)-n\fp_\varphi(\GG)}\sum_{x\in E_{\e,n}}\eee^{G_n(x)
+ G_n'(x)}.
\end{align*}
Taking the supremum over all $(\e,n)$-separated sets, repeating the above argument, and 
using expression~\eqref{2.4} for $\fp_\varphi(\GG+\GG')$, we obtain the desired result.
\hfill\qed 	
\end{proof}

\begin{lemma}
Let~$\P\in\PP(M)$ be a weak  Gibbs measure. Then Conditions~\hCone 
and~\hCtwo hold for $\P_n = \P$.
\end{lemma}

\begin{proof}{} \hCone follows from the special case $\GG'=\{S_n V\}$ in 
Lemma~\ref{lem:WGPress}. \hCtwo with $m_\delta(n)\equiv0$ follows from  
Definition~\ref{d1.16}. 
\hfill\qed 	
\end{proof}

\subsection{Proof of the LDP for empirical measures}
\label{subs:LDPOM}

We first give the  proof of Theorem~\ref{appt1.17} (1), which is completed in the following two steps.

\paragraph{Step 1: LD upper bound.}

\begin{proposition} \label{p1.10}
If Condition~\hCone holds, then
\begin{equation} \label{2.19}
\limsup_{n\to\infty}\frac1n\log\IP_n\{\mu_n^\argdot\in F\}
\le-\inf_{\Q\in F}{\I}(\Q)
\end{equation}
holds for any closed set $F\subset\PP(M)$.
\end{proposition}

\begin{proof}{}
It is a well-known fact (see~\cite[Theorem~2.1]{kifer-1990} and
also~\cite[Section 4.5.1]{DZ2000}) that the existence of limit~\eqref{2.03first} implies
inequality~\eqref{2.19} with a rate function~$\I$ given by the Legendre transform of 
$\fp_\varphi(\GG_V)-\fp_\varphi(\GG)$ with respect to $V$, \ie by~\eqref{4.011}.
\hfill\qed
\end{proof}

\paragraph{Step 2: LD lower bound.}

We need to prove that, for any open set $O\subset\PP(M)$, 
\[
\liminf_{n\to\infty}\frac1n\log\IP_n\{\mu_n^\argdot\in O\}\ge-\inf_{\Q\in O}{\I}(\Q). 
\]
This inequality will be established if we prove that, for any $\Q\in O$, 
\begin{equation} \label{2.20}
\liminf_{n\to\infty}\frac1n\log\IP_n\{\mu_n^\argdot\in O\}\ge-\I(\Q). 
\end{equation}
Moreover, we only need to consider $\Q\in\PP_\varphi(M)$,
since otherwise $\I(\Q)=+\infty$ and~\eqref{2.20} is trivially satisfied.
To prove~\eqref{2.20} for $\Q\in\PP_\varphi(M)$, we follow a strategy that goes 
back to~\cite{FO-1988}, see also~\cite{EKW-1994,PS-2005}, and consider first the special 
case $\Q\in\EE_\varphi(M)$.

\begin{proposition} \label{appp1.11}
Assume Condition~\hCtwo. Then inequality~\eqref{2.20} holds 
for any open set $O\subset \PP(M)$ and any $\Q\in O\cap\EE_\varphi(M)$.  
\end{proposition}
\begin{proof}{}
Fix $O\subset\PP(M)$ and $\Q\in O\cap\EE_\varphi(M)$. Given $V_1,\dots,V_m\in C(M)$ and 
$\e>0$ we set 
\begin{equation*}
R_\e=R_\e(V_1,\dots,V_m)=\bigcap_{j=1}^m
\bigl\{\Q'\in\PP(M):|\langle V_j,\Q'\rangle-\langle V_j,\Q\rangle|<\e\bigr\}.
\end{equation*}
Since~$O$ is open and $\Q \in O$, we can find finitely many functions 
$V_1,\dots,V_m\in C(M)$ and a number~$\e_0>0$ such that $R_{2\e_0}\subset O$.  Let  
$\e\in]0,\e_0[$ and let $X_n^\e$ be  the set of points $x\in M$ such that 
\[
|n^{-1}G_n(x)-\GG(\Q)|<\e, \quad 
|n^{-1}S_n V_j(x)-\langle V_j,\Q\rangle|<\e\quad
\text{for }j=1,\dots,m. 
\]
Note that
\begin{equation*}
X_n^\e \subset X_n^{2\e} \subset \{x\in M:\mu_n^x\in R_{2\e}\}
\subset \{x\in M:\mu_n^x\in O\}.
\end{equation*}
Since $\Q$ is ergodic, it follows from Lemma~\ref{lem:ergoasymadd} that 
\begin{equation} \label{app2.22}
\lim_{n\to\infty}\Q(X_n^\e)=1.
\end{equation}
For $\delta>0$ and $n\ge n_0(\delta)$ we define
\begin{equation*}
Y_n^\e(\delta)=\bigl\{x\in M:\Q(B_{n-m_{\delta/2}(n)}(x,\delta))
\le\eee^{-n(h_\varphi(\Q)-\e)}\bigr\}.
\end{equation*}

Using the fact that
\[
\lim_{\delta\downarrow0}\liminf_{n\to\infty}\frac{n-m_{\delta/2}(n)}{n}
=\lim_{\delta\downarrow0}\limsup_{n\to\infty}\frac{n-m_{\delta/2}(n)}{n}=1,
\]
and invoking the ergodicity of~$\Q$, the Brin--Katok local entropy formula 
(see~\cite{BK-1983}) implies that
\[
\lim_{\delta\downarrow0}\limsup_{n\to\infty}
\frac1n\log\Q\bigl(B_{n-m_{\delta/2}(n)}(x,\delta)\bigr)=-h_\varphi(\Q)
\quad\text{for $\Q$-almost every }x\in M.
\]
Combining this with a simple measure-theoretic argument (similar to the one 
used to prove Egorov's theorem), we see that
\begin{equation} \label{app2.23}
\lim_{\delta\downarrow0}\liminf_{n\to\infty}\Q\bigl(Y_n^\e(\delta)\bigl)=1. 
\end{equation}
It follows from~\eqref{app2.22} and~\eqref{app2.23} that for all small enough 
$\delta>0$ there is an integer $n_1(\delta)\ge n_0(\delta)$ such that 
\begin{equation}\label{app2.25}
\Q\bigl(X_n^\e\cap Y_n^\e(\delta)\bigr)\ge\tfrac12\quad\text{for all }n\ge n_1.
\end{equation}
Moreover, by~\eqref{4.010ss} (applied to $G_n$ and to the potential sequences 
$\{S_n V_j\}$) and~\eqref{eq:conditionKnMdelta}, we can assume, by possibly 
decreasing $\delta$ and increasing $n_1$, that for all $n\ge n_1$, 
\begin{gather}
\sup_{x\in M}\sup_{y,z\in B_{n-m_{\delta/2}(n)}(x,\delta/2)}
\frac1n\bigl|G_{n}(y)-G_{n}(z)\bigr|<\e,\label{eq:condxmbn}\\
\sup_{x\in M}\sup_{y,z\in B_{n-m_{\delta/2}(n)}(x,\delta/2)}
\frac1n\bigl|S_nV_j(y)-S_n V_j(z)\bigr|<\e,
\quad\text{for }j=1,\dots,m,\label{eq:condxmbnSn}\\
\frac1n\log K_n(\delta/2)< \e.\label{eq:conditionKmeps}
\end{gather}
Suppose that, for sufficiently large~$n$, we have constructed points 
$x_1,\dots,x_{r_n}\in X_n^\e\cap Y_n^\e(\delta)$ such that the balls 
$B_{n-m_{\delta/2}(n)}(x_i,\delta/2)$ are pairwise disjoint, and
\begin{gather} 
\bigcup_{i=1}^{r_n}B_{n-m_{\delta/2}(n)}(x_i,\delta/2)
\subset X_n^{2\e}, \qquad 
X_n^\e\cap Y_n^\e(\delta)\subset
\bigcup_{i=1}^{r_n}B_{n-m_{\delta/2}(n)}(x_i,\delta). 
\label{app2.26}
\end{gather}
In this case, we can write 
\begin{align} 
\IP_n\{\mu_n^\argdot\in O\}\ge \IP_n(X_n^{2\e})
&\ge \sum_{i=1}^{r_n}\IP_n\bigl(B_{n-m_{\delta/2}(n)}(x_i,\delta/2)\bigr)\notag\\
&=\sum_{i=1}^{r_n}
\frac{\IP_n\bigl(B_{n-m_{\delta/2}(n)}(x_i,\delta/2)\bigr)}
{\Q\bigl(B_{n-m_{\delta/2}(n)}(x_i,\delta)\bigr)}\,
\Q\bigl(B_{n-m_{\delta/2}(n)}(x_i,\delta)\bigr).
\label{app1.27}
\end{align}
Since $x_i\in Y_n^\e(\delta)$, we have
\begin{equation} \label{eq:app2.31}
\Q\bigl(B_{n-m_{\delta/2}(n)}(x_i,\delta)\bigr)\le\eee^{-n(h_\varphi(\Q)-\e)}
\end{equation}
for $n\ge n_0$. Moreover, \eqref{eq:applowerassum} and  $x_i\in X_n^\e$ imply that
\begin{equation}
\begin{split}
\IP_n\bigl(B_{n-m_{\delta/2}(n)}(x_i,\delta/2)\bigr) &\geq K_n(\delta/2)^{-1}\eee^{ G_n(x_i) - n \fp_\varphi(\GG)}\\
& \geq K_n(\delta/2)^{-1}\eee^{ n \GG(\Q) - n\e - n\fp_\varphi(\GG)}.
\end{split}\label{eq:app2.32}
\end{equation} 
Thus, using \eqref{app1.27}, \eqref{eq:app2.31}, \eqref{eq:app2.32}, the second 
inclusion in \eqref{app2.26},  and \eqref{app2.25}, we derive
\begin{align*} 
\IP_n\{\mu_n^\argdot\in O\}
&\ge\frac12 K_n(\delta/2)^{-1}\eee^{n(\GG(\Q)+h_\varphi(\Q))
-2n\e-n\fp_\varphi(\GG) }\\
&\ge\frac12 K_n(\delta/2)^{-1}\eee^{-n\I(\Q)-2n\e },
\end{align*}
where the second inequality follows from~\eqref{4.012}. Together 
with~\eqref{eq:conditionKmeps} this gives
\[
\liminf_{n\to\infty}n^{-1}\log\IP_n\{\mu_n^\argdot\in O\}\ge-\I(\Q)-2\e. 
\]
Since~$\e\in]0,\e_0[$ can be chosen arbitrarily small, \eqref{2.20} follows. 

It remains to construct points 
$x_1,\dots,x_{r_n}\in X_n^\e\cap Y_n^\e(\delta)$ such that the balls 
$B_{n-m_{\delta/2}(n)}(x_i,\delta/2)$ are disjoint and~\eqref{app2.26} holds. 

First, it follows from~\eqref{eq:condxmbn} and~\eqref{eq:condxmbnSn} that for 
all $x\in M$ and $y\in B_{n-m_{\delta/2}(n)}(x,\delta/2)$,
\begin{gather*}
\bigl|n^{-1}G_n(y)-\GG(\Q)\bigr|\leq\e+\bigl|n^{-1}G_n(x)-\GG(\Q)\bigr|
< 2\e,\\[2mm]
\bigl|n^{-1}S_nV_j(y)-\langle V_j,\Q\rangle\bigr|\leq\e
+\bigl|n^{-1}S_nV_j(x)-\langle V_j,\Q\rangle\bigr|< 2\e
\quad \text{for } j=1, \dots, m,
\end{gather*}
and so
\beq
x\in X_n^\e\Longrightarrow B_{n-m_{\delta/2}(n)}(x,\delta/2)\subset X_n^{2\e}
\label{app2.33}
\eeq
for all $n\geq n_1$. Now let
$\mathfrak{B}=\{B_{n-m_{\delta/2}(n)}(x_i,\delta/2):1\le i\le r_n\}$ be any
maximal\footnote{Such a maximal collection exists, since~\eqref{eq:applowerassum} 
gives an absolute upper bound on $r_n$.} collection of disjoint balls included 
in~$X_n^{2\e}$ such that
$x_i\in X_n^\e\cap Y_n^\e(\delta)$. The first inclusion
in~\eqref{app2.26} follows from~\eqref{app2.33}. To prove the second one, suppose 
that $x_*\in X_n^\e\cap Y_n^\e(\delta)$ does
not belong to any of the balls $B_{n-m_{\delta/2}(n)}(x_i,\delta)$. Then
$B_{n-m_{\delta/2}(n)}(x_*,\delta/2)$ does not intersect the balls
in~$\mathfrak{B}$ and, by~\eqref{app2.33}, is included in $X_n^{2\e}$, and  the
collection~$\mathfrak{B}$ is not maximal. This completes the proof of the 
proposition.
\hfill\qed\end{proof}

The following proposition completes the proof of Part~(1) of 
Theorem~\ref{appt1.17}. 

\begin{proposition}
If, in addition to the hypotheses of Proposition~\ref{appp1.11}, 
Condition~\hS holds, then inequality~\eqref{2.20} holds for any 
open set $O\subset \PP(M)$ and any $\Q\in O\cap\PP_\varphi(M)$. 
\end{proposition}

\begin{proof}{}
Let $\Q\in O\cap\cP_\varphi(M)$. By Proposition~\ref{Prop:EntropyDensity}, there 
exists a sequence $\{\Q^{(m)}\}\subset\EE_\varphi(M)$ such that
\[
\Q^{(m)}\rightharpoonup\Q,\qquad h_\varphi(\Q^{(m)})\to h_\varphi(\Q),
\]
as $m\to\infty$. In this case, in view of~\eqref{4.012} and the continuity 
assertion in Lemma~\ref{lem:limiteGGPP}, we have 
\begin{equation} \label{2.34}
\I(\Q^{(m)})=-\GG(\Q^{(m)})-h_\varphi(\Q^{(m)})+\fp_\varphi(\GG) 
\to -\GG(\Q)-h_\varphi(\Q)+\fp_\varphi(\GG) 
=\I(\Q) 
\end{equation}
as $m\rightarrow \infty$.
Since $O$ is open, $\Q^{(m)}\in O$ for large enough $m$, and by 
Proposition~\ref{appp1.11}, we have 
\[
\liminf_{n\to\infty}n^{-1}\log\IP_n\{\mu_n^\argdot\in O\}\ge -{\I}(\Q^{(m)}).
\]
Passing to the limit $m\to\infty$ and using~\eqref{2.34} we 
obtain inequality~\eqref{2.20}. 
\hfill\qed
\end{proof}

\subsection{Proof of the LDP for asymptotically additive sequences of functions}
\label{ss:LDPlev1}

Part (2) of Theorem~\ref{appt1.17} is a special case of Theorem~4.2.23 
in~\cite{DZ2000}.\footnote{In the notation therein, ${\cal X}=\cP(M)$, 
${\cal Y}=\rr$, $f_m(\Q)=\langle V^{(m)},\Q \rangle$ and $f(\Q)$ is defined by 
$\VV(\Q)$ when $\Q\in \cP_\varphi(M)$ and arbitrarily when 
$\Q\in\cP(M)\setminus\cP_\varphi(M)$.} For the reader's convenience, we outline 
the proof in our case.

Let  $\cV=\{V_n\}\in{\cal A}(M)$ with approximating sequence $\{V^{(k)}\}$. We 
define the random variables $\xi_n=\frac1n V_n$ and 
\[
\xi_n^k=\frac1n S_n V^{(k)}=\langle V^{(k)},\mu_n^\argdot\rangle,
\]
and consider them under the law~$\P_n$. By the definition of $V^{(k)}$ we have
\begin{equation} \label{1.121}
\lim_{k\to\infty}\limsup_{n\to\infty}\|\xi_n^k-\xi_n\|_\infty=0. 
\end{equation}
By the contraction principle~\cite[Theorem 4.2.1]{DZ2000}, for each $k$ the 
family $\{\xi_n^k\}_{n\geq 1}$ satisfies  the LDP with the good convex rate function 
\begin{equation}\label{eq:formulecontractionIk}
I_k(s)=\inf\{\I(\Q):\Q\in\cP_\varphi(M),\langle V^{(k)},\Q\rangle=s\}.
\end{equation}
We now show that $\xi_n$ satisfies the LDP with the rate function
\begin{equation*}
I(s)=\sup_{\delta>0}\liminf_{k\to\infty}\inf_{z\in B(s, \delta)} I_k(z) .
\end{equation*}

It is immediate that $I$ is lower semicontinuous. By~\eqref{1.121},  
for all $\delta > 0$,
\begin{align*}
\limsup_{n\to\infty} \frac 1 n \log  \P_n\{ \xi_n \in B(s, \delta)\}
&\leq-\liminf_{k\to\infty}\inf_{y\in  B(s, 2\delta)}I_k(y),\\
\liminf_{n\to\infty} \frac 1 n \log  \P_n\{ \xi_n \in B(s, \delta)\}
&\geq-\liminf_{k\to\infty}\inf_{y\in  B(s, \delta/2)}I_k(y).
\end{align*}
Hence we obtain
\begin{equation*}
\sup_{\delta > 0}\limsup_{n\to\infty} 
\frac 1 n \log  \P_n\{\xi_n \in B(s, \delta)\}  
= \sup_{\delta > 0}\liminf_{n\to\infty} 
\frac 1 n \log  \P_n\{ \xi_n \in B(s, \delta)\} = I(s).
\end{equation*}
A standard argument~\cite[Theorem 4.1.11]{DZ2000} implies that~$\xi_n$ satisfies 
the weak LDP with rate function~$I$. Since the family~$\{\xi_n\}$ is bounded 
(recall Remark~\ref{rem:equivrel}), $\xi_n$ actually satisfies the full LDP, 
and~$I$ is a good rate function. It remains to show that $I(s)=J(s)$, where 
\begin{equation*}
J(s)=\inf\{\I(\Q):\Q\in\cP_\varphi(M),\VV(\Q) = s\}.
\end{equation*}
Note that $J$  is lower semicontinuous, since $\Q \mapsto \VV(\Q)$ is continuous 
on $\cP_\varphi(M)$ by Lemma~\ref{lem:limiteGGPP}, and is obviously convex.
It follows from~\eqref{eq:formulecontractionIk} that 
\begin{equation*}
I(s)=\sup_{\delta > 0}\liminf_{k\to\infty}\inf\{\I(\Q):
\Q\in\cP_\varphi(M),\langle V^{(k)},\Q\rangle\in B(s,\delta)\}, 
\end{equation*}
while the lower semicontinuity of $J$ gives
\begin{equation*}
J(s)=\sup_{\delta>0}\inf_{y\in B(s,\delta)}J(s)
=\sup_{\delta>0}\inf\{\I(\Q):\Q\in\cP_\varphi(M),\VV(\Q)\in B(s,\delta)\}.
\end{equation*}
Using that $\langle V^{(k)},\Q\rangle\to\VV(\Q)$ uniformly on $\cP_\varphi(M)$ 
(recall~\eqref{4.04b}), we derive that $J(s)=I(s)$. This completes the proof of 
Part~(2) of Theorem~\ref{appt1.17}.

\section{Conditions for asymptotic additivity}
\label{sec:charac}

In this section we give some necessary and sufficient conditions for a potential 
sequence to be asymptotically additive. 

A sequence $\GG=\{G_n\}\subset B(M)$ is said to have \textit{tempered 
variation\/} if
\begin{equation}\label{eq:mildvar3}
\lim_{\e \downarrow 0}\limsup_{n\to\infty}\sup_{x\in M}\sup_{y,z\in B_n(x,\e)} 
\frac 1 n\bigl|G_n(y)-G_n(z)\bigr| = 0.	
\end{equation}
We have shown in Lemma~\ref{lem:approxGmn} that asymptotic additivity 
implies~\eqref{eq:mildvar3}. Below, we shall sometimes take~\eqref{eq:mildvar3} 
as an assumption (along with others), in order to obtain asymptotic additivity.

We recall that $\{G_n\}\subset B(M)$ is called \textit{weakly almost additive\/} 
if for all $n,m\geq 1$, 
\begin{equation}\label{sub-ad2}
-C_m + G_m +G_n\circ\varphi^m \leq G_{m+n}\leq C_m + G_m +G_n\circ\varphi^m,
\end{equation}
where $\lim_{n\rightarrow \infty}n^{-1} C_n=0$.

The main result of this section is
\begin{theorem}\label{prop:characterizationAA}
If $\GG=\{G_n\}_{n\geq 1}\subset B(M)$ satisfies any of the following conditions, 
then $\GG$ is asymptotically additive.
\begin{enumerate}[~~(1)]
\item $G_n = S_nG$ for each $n$, with $G\in C(M)$.
\item $G_n = S_nG$ for each $n$, with $G\in B(M)$, and $\GG$ has tempered variation.
\item $\GG$ is weakly almost additive and $G_n \in C(M)$ for each $n$.
\item $\GG$ is weakly  almost additive and has tempered variation.
\end{enumerate}
Moreover, if $\GG \subset C(M)$, then the following assertion is equivalent to asymptotic additivity of $\GG$.
\begin{enumerate}[~~(1)]
\setcounter{enumi}{4}
\item $\GG$ satisfies
 \begin{equation}\label{eq:WTSGk25}
\lim_{k\to\infty}\limsup_{n\to\infty}\,n^{-1} \bigl\|G_n-k^{-1} S_n G_k\bigr\|_\infty = 0.
\end{equation} 
\end{enumerate}
Finally, for $\GG \subset B(M)$, each of the following assertions is equivalent to asymptotic additivity of $\GG$.
\begin{enumerate}[~~(1)]
\setcounter{enumi}{5}
\item $\GG$ has tempered variation and satisfies~\eqref{eq:WTSGk25}.
\item $\GG$ has tempered variation and there exists a sequence $ \{G^{(k)}\} \subset B(M)$ such that
 \begin{equation}\label{eq:WTSGk2}
\lim_{k\to\infty}\limsup_{n\to\infty}\,n^{-1} \bigl\|G_n-S_n  G^{(k)}\bigr\|_\infty = 0.
\end{equation}
\item There exists a sequence $\{G^{(k)}\} \subset B(M)$ such that \eqref{eq:WTSGk2} holds, and 
 \begin{equation}\label{eq:mildvareachk}
\lim_{k\to\infty}\lim_{\e \downarrow 0}\limsup_{n\to\infty}\sup_{x\in M}\sup_{y,z\in B_n(x,\e)} 
\frac 1 n\bigl|S_n  G^{(k)}(y)-S_n   G^{(k)}(z)\bigr| = 0.	
\end{equation}
\end{enumerate}
\end{theorem}

Let us mention that various partial results contained in 
Theorem~\ref{prop:characterizationAA} were known earlier (e.g., see the 
papers~\cite{barreira-2006,zhao_asymptotically_2011,Bar2011} and the references 
therein). However, the equivalence relationships stated above seem to be new.  

We start with the following lemma, which was established 
in~\cite[Proposition A.5]{FH-2010} 
and~\cite[Proposition 2.1]{zhao_asymptotically_2011} in the almost additive  
case, that is, when~$\{C_m\}$ in~\eqref{sub-ad2} is a constant sequence.
\begin{lemma}\label{p4.30} 
Assume that $\{G_n\} \subset B(M)$ is weakly almost additive. 
Then~\eqref{eq:WTSGk25} holds.
\end{lemma}
\begin{proof}{}
Given two positive integers~$n$ and~$k$, we write~$n_k$ for the integer part 
of~$n/k$ and, for a function $V$, we let
\[
S_n^k V=\sum_{r=0}^{n-1} V\circ\varphi^{rk}. 
\]
Suppose that for any~$\e>0$ we can find~$k_\e\ge1$ such that 
\begin{equation} \label{4.38}
\limsup_{n\to\infty}\,
\bigl\|n^{-1}G_n-n^{-1}S_{n_k}^{\,k}G_k\bigr\|_\infty\le\e
\quad\text{for }k\ge k_\e. 
\end{equation}
For a fixed $k\ge k_\e$ and any $\ell\in\llbracket 1,k-1\rrbracket$, 
replacing~$x$ by~$\varphi^\ell(x)$ in~\eqref{4.38} and using an elementary 
estimate for the ergodic average, we derive 
\[
\limsup_{n\to\infty}\,
\bigl\|n^{-1}G_n\circ\varphi^\ell
-n^{-1}S_{(n+\ell)_k}^{\,k}G_k\circ\varphi^\ell\bigr\|_\infty\le\e
\quad\mbox{for $k\ge k_\e$}. 
\]
Combining this with the relation
\[
\lim_{n\to\infty}\,
\bigl\|(n+\ell)^{-1}G_{n+\ell}-n^{-1}G_n\circ\varphi^\ell\bigr\|_\infty=0,
\] 
which follows from~\eqref{sub-ad2}, we obtain
\begin{equation} \label{4.39}
\limsup_{n\to\infty}\,
\bigl\|n^{-1}G_n-n^{-1}S_{n_k}^{\,k} G_k\circ\varphi^\ell\bigr\|_\infty\le\e
\quad\mbox{for $k\ge k_\e$}. 
\end{equation}
Now note that 
\[
k^{-1}\sum_{\ell=0}^{k-1}S_{n_k}^{\,k} G_k\circ\varphi^\ell=k^{-1}S_{kn_k}G_k. 
\]
Comparing  with~\eqref{4.39}  we get that for  $k\ge k_\e$,
\[
\limsup_{n\to\infty}\,
\bigl\|n^{-1}G_n-n^{-1}k^{-1}S_{kn_k}G_k\bigr\|_\infty
=\limsup_{n\to\infty}\,
\bigl\|n^{-1}G_n-n^{-1}k^{-1}S_{n}G_k\bigr\|_\infty\le\e. 
\]
Since~$\e>0$ is arbitrary, the relation~\eqref{eq:WTSGk25} follows. 

We now prove~\eqref{4.38}. Let us fix an integer $k\ge1$ and write, for 
$n$ large,  $n=kn_k+\ell$, where $0\le \ell\le k-1$. Applying 
inequality~\eqref{sub-ad2} consecutively~$n_k$ times, we derive
\[
G_n\le k_n 
C_k+\sum_{r=0}^{n_k-1}G_k\circ\varphi^{rk}+G_\ell\circ\varphi^{kn_k}. 
\]
This gives 
\[
\limsup_{n\to\infty}\sup_{x\in M}
\bigl(n^{-1}G_n(x) -n^{-1}S_{n_k}^{\,k}G_k(x)\bigr)\le k^{-1}C_k.
\]
Replacing~$G_n$ by~$-G_n$, we derive
\[
\liminf_{n\to\infty}\inf_{x\in M}
\bigl(n^{-1}G_n(x)-n^{-1}S_{n_k}^{\,k}(G_k)(x)\bigr)\ge -k^{-1}C_k.
\]
Combining the last two inequalities and recalling that 
$n^{-1} C_n\to 0$, we arrive at~\eqref{4.38}. 
\hfill\qed
\end{proof}

\begin{lemma}\label{lem:techGnGj}
Let $\GG=\{G_n\}\subset B(M)$ be such that there exists 
$\{G^{(k)}\} \subset B(M)$ satisfying~\eqref{eq:WTSGk2}. 
Then~\eqref{eq:WTSGk25} holds.
\end{lemma}
\begin{proof}{} 
For all $j\geq 1$ we have
\begin{align*}
\|G_n-k^{-1}S_n G_k\|_\infty 
\leq\|G_n-S_n G^{(j)}\|_\infty
&+\|S_n G^{(j)}-k^{-1}S_n S_k G^{(j)}\|_\infty\\
&+k^{-1}\|S_n S_k G^{(j)}-S_n G_k\|_\infty.
\end{align*}

Applying Lemma~\ref{p4.30} to $\{S_n G^{(j)}\}$ we obtain that for each 
$j$,
\begin{equation*}
\lim_{k\to\infty}\limsup_{n\to\infty}\frac1n
\|S_n G^{(j)}-k^{-1}S_n S_k G^{(j)}\|_\infty = 0.
\end{equation*}
Fix now $\e > 0$. If~$j$ is large enough, then by~\eqref{eq:WTSGk2} we 
have 
\begin{align*}
\limsup_{n\to\infty} \frac  1 n	\|G_n - S_n G^{(j)}\|_\infty 
&\leq \epsilon,\\
\limsup_{k\to\infty}\limsup_{n\to\infty}\frac1{kn}
\|S_n S_k G^{(j)}-S_n G_k\|_\infty  
&\leq \limsup_{k\to\infty}\frac1{k}\| S_k G^{(j)}-G_k\|_\infty\leq\e.
\end{align*}
We thus obtain
\begin{equation*}
\limsup_{k\to\infty}\limsup_{n\to\infty}\frac1n
\|G_n-k^{-1}S_n G_k\|_\infty \leq 2\epsilon.
\end{equation*}
Since $\epsilon$ is arbitrary, this completes the proof.\hfill\qed
\end{proof}

\begin{lemma}\label{lem:regularizationBB}
Let $f\in B(M)$ be such that for some fixed $n,\e,\alpha$ we have
\begin{equation*}
\sup_{x\in M} \sup_{y\in B_n(x, \e)} |f(y)-f(x)|	\leq \alpha.
\end{equation*}
Then there exists a continuous function $g$ such that 
$\|f-g\|_\infty	\leq \alpha$.
\end{lemma}

\begin{proof}{}
Let $E = \{x_1, x_2, \dots, x_r\}$ be any finite $(n,\e)$-spanning set. 
Let $\rho_1, \dots, \rho_r$ be a partition of unity subordinated to the 
collection $\{B_n(x_i, \e): i=1, \dots, r\}$ (\ie $\rho_i $ is 
continuous, vanishes outside $B_n(x_i, \e)$, and $\sum_i \rho_i =1$).
We claim that the continuous function
\begin{equation*}
g = \sum_{i=1}^r \rho_i f(x_i)	
\end{equation*}
satisfies the required properties. 
Let  $x\in M$ and let $J\subset \{1, \dots, r\}$ be the largest set such 
that $x\in \bigcap_{j\in J} B_n(x_j, \e)$. We then  have
\begin{equation*}
|g(x) -f(x)| = |	\sum_{j\in J} \rho_j (x)f(x_j) -f(x)| 
= |\sum_{j\in J} \rho_j(x) (f(x_j) -f(x))| \leq \alpha,
\end{equation*}
and the result follows.\hfill\qed
\end{proof}

\begin{proof}{ of Proposition~\ref{prop:characterizationAA}}
In case~(1) we can obviously choose $G^{(k)} = G$ as an approximating 
sequence for $\GG$. Next, (2) is a special case of~(4) with 
$C_n \equiv 0$. In case~(3), we obtain by Lemma~\ref{p4.30} 
that~\eqref{eq:WTSGk25} holds so that we are in case~(5).
In case~(4), \eqref{eq:WTSGk25} also holds by Lemma~\ref{p4.30}, and 
hence we find ourselves in the case~(6).

That~(5) implies asymptotic additivity is immediate, with 
$G^{(k)} = k^{-1} G_k$ as an approximating sequence. The reverse 
implication follows immediately from Lemma~\ref{lem:techGnGj} applied to 
$G_n$ and any approximating sequence $\{G^{(k)}\}$.

We now prove that~(6) implies asymptotic additivity. First, it follows 
from~\eqref{eq:mildvar3} that there exists a sequence 
$\{\e_k\}_{k\geq 1}$ such that
\begin{equation}\label{eq:subseqepsk}
\lim_{k\to \infty} \sup_{x\in M}\sup_{y\in B_{k}(x,\e_k)}\frac1k
|G_{k}(x)-G_{k}(y)| = 0.
\end{equation}
Indeed, by~\eqref{eq:mildvar3}, for each $\ell \in \mathbb N$, there 
exists $\overline \e(\ell)$ and $k_0(\ell)$ such that for all 
$k \geq k_0(\ell)$ we have $\sup_{x\in M} \sup_{y\in B_{k}(x,\overline \e(\ell))}\frac 1 {k}|G_{k}(x)-G_{k}(y)| \leq \ell^{-1}$.  Let  $\{\ell_k\}$  be such that $\ell_k\to\infty$ and $k\geq k_0(\ell_k)$ for all $k$ large enough. Setting  $\e_k = \overline \e(\ell_k)$ 
establishes~\eqref{eq:subseqepsk}. 

Let $G^{(k)}$ be  the regularization of $\frac 1 {k} G_{k}$ obtained in 
Lemma~\ref{lem:regularizationBB} with respect to the Bowen balls 
$B_k(x, \e_k)$, with $\e_k$ as in \eqref{eq:subseqepsk}.  The function 
$G^{(k)}$ is continuous and
\begin{equation}\label{eq:limitGkGpk}
\lim_{k\to\infty}\|k^{-1} G_{k} - G^{(k)}\|_\infty = 0.
\end{equation}
Since
\begin{align*}
\frac 1 n\|G_n - S_n G^{(k)}\|_\infty& \leq \frac  1n
\|G_n - k^{-1} S_n G_{k}\|_\infty +  \frac 1 {n}\left\|  
S_n \left[k^{-1} G_{k} -  G^{(k)}\right ]\right\|_\infty,
\end{align*}
relations  \eqref{eq:WTSGk25} and~\eqref{eq:limitGkGpk} give  that $G_n$ 
is asymptotically additive.

Next,  Lemma~\ref{lem:techGnGj} immediately implies that~(7) is a 
special case of~(6). Finally, assuming~\eqref{eq:WTSGk2}, it is easy to 
see that~\eqref{eq:mildvar3} and~\eqref{eq:mildvareachk} are equivalent. 
Thus, (7) and~(8) are equivalent.

We have shown that (8) $\iff$ (7) $\Longrightarrow$ (6) 
$\Longrightarrow \GG \in \cA(M)$. Since by Lemma~\ref{lem:approxGmn} 
asymptotic additivity implies~(7), the statements (6), (7), (8) are all 
equivalent to $\GG \in \cA(M)$.
\hfill\qed\end{proof}

\begin{remark}\label{rem:remcharac} 
Note that by the characterization given in (5), if $G_n\in C(M)$ for 
all~$n$, then the approximating sequence can be chosen to be 
$G^{(k)} = k^{-1} G_k$. Moreover, the proof gives that when these 
functions are not continuous, $G^{(k)}$  can be chosen as a 
regularization of $k^{-1} G_k$. By the finiteness 
of~\eqref{eq:nm1Gnfini}, this specific choice of $G^{(k)}$ satisfies
$\sup_{k\geq 1}\|G^{(k)}\|_\infty < \infty$ (this is not true of all
approximating sequences). Finally, if $\GG = \{G_n\} \subset B(M)$ is
asymptotically additive, then there exists an asymptotically additive 
potential sequence $\{G_n'\}  \subset C(M)$ in the same class as $\GG$ 
in the sense of Remark~\ref{rem:equivrel}, \ie such that 
$\limsup_{n\to\infty} n^{-1}\|G_n-G_n'\|_\infty =0$. Indeed, it 
suffices to take an approximating sequence $\{G^{(k)}\} \subset C(M)$ 
for $\GG$, and then to define $G_n' = S_n G^{(k_n)}$ for some 
well-chosen sequence $k_n\to\infty$ (which is obtained with an argument
similar to that leading to~\eqref{eq:subseqepsk}).
\end{remark}

\begin{remark}\label{rem:openQaa} 
The reader may check that $G_n = \log n$ gives a sequence which is weakly almost
additive but not almost additive.\footnote{We recall that a sequence is almost
additive if~\eqref{sub-ad2} holds with $C_n$ independent of $n$.} Moreover,
choosing $G_n = \sqrt{n}$ when $n$ is even, and $G_n = 0$ when $n$ is odd, gives
a sequence which is asymptotically additive but not weakly almost additive. We
note that these two potential sequences are actually equivalent (in the sense of
Remark~\ref{rem:equivrel}) to the potential which is identically zero. It
remains an open question whether one can find an asymptotically additive
potential $\GG$ such that there is no additive potential in the same equivalence
class. { {\bf Note added in Feb.~2026}: the answer to this question
was proved to be negative in \cite{cuneo_asympt_2019}.
}

\end{remark}

\subsection*{Frequently used notation}
\addcontentsline{toc}{section}{Frequently used notation}

\begin{longtable}{p{3.2cm}p{12.cm}}
\textbf{(S)} & specification property, page~\pageref{Cond:S}\\
\textbf{(WPS)} & weak periodic specification property, page~\pageref{Cond:WPS}\\
\textbf{(USCE)} & upper semicontinuity of entropy, page~\pageref{usce}\\
\textbf{(PAP)} & periodic approximation of pressure, page~\pageref{pap}\\ 
\textbf{(C)} & $\varphi$ is continuous, page~\pageref{C}\\
\textbf{(H)} & $\varphi$ is a homeomorphism, page~\pageref{H}\\
\textbf{($\boldsymbol{\mathcal{C}}$-Commutation)} & commutation hypothesis, page~\pageref{commutation}\\
\textbf{($\boldsymbol{\mathcal{R}}$-Reversal)} & reversal hypothesis, page~\pageref{reversal} \\
$M$ & compact metric space\\
$\varphi$ & continuous mapping of the space~$M$ into itself\\
$M_n$ & set of fixed points of the mapping~$\varphi^n$\\
$C(M)$ & space of continuous functions $V:M\to\R$ with the supremum norm~$\|\cdot\|_\infty$\\
$B(M)$ & space of bounded measurable functions $V:M\to\R$ with the norm~$\|\cdot\|_\infty$\\
$\aA(M)$ & space of asymptotically additive sequences of functions, page~\pageref{eq:defasymadd}\\
$\PP(M)$ & set of Borel probability measures on $M$, with the topology of weak convergence and the corresponding Borel $\sigma$-algebra\\
$\PP_\varphi(M)$ & set of invariant measures for a mapping~$\varphi$\\
$\EE_\varphi(M)$ & set of ergodic invariant measures for~$\varphi$\\
$h_\textrm{Top}(\varphi)$ & topological entropy of~$\varphi$\\
$h_\varphi(\Q)$ & Kolmogorov--Sinai entropy of~$\varphi$\\
$\fp_\varphi(\GG)$ & topological pressure of $\GG\in\aA(M)$ with respect to $\varphi$, page~\pageref{sec:TP}\\
$\mu_n^x$ & empirical measures, page~\pageref{oc-me}\\
$\sigma_n$ & the entropy production in time~$n$, page~\pageref{4.31}
\end{longtable}

\addcontentsline{toc}{section}{References}

\newcommand{\etalchar}[1]{$^{#1}$}
\def\polhk#1{\setbox0=\hbox{#1}{\ooalign{\hidewidth \lower1.5ex\hbox{`}\hidewidth\crcr\unhbox0}}} \def\pre{{Phys. Rev. E\ }}\def\polhk#1{\setbox0=\hbox{#1}{\ooalign{\hidewidth \lower1.5ex\hbox{`}\hidewidth\crcr\unhbox0}}}
\providecommand{\bysame}{\leavevmode \hbox to3em{\hrulefill}\thinspace}
\providecommand{\og}{``}
\providecommand{\fg}{''}
\providecommand{\smfandname}{and}
\providecommand{\smfedsname}{eds.}
\providecommand{\smfedname}{ed.}
\providecommand{\smfmastersthesisname}{Master Thesis}
\providecommand{\smfphdthesisname}{Thesis}

\end{document}